\documentclass[acmsmall,screen,nonacm]{acmart}

\usepackage[utf8]{inputenc}
\usepackage[prologue,dvipsnames]{xcolor}
\newcommand{\bind}{\mathrel{\scalebox{0.5}[1]{$>\!>=$}}}
\newcommand{\hide}[1]{}

\usepackage{stmaryrd}
\usepackage{prooftree}
\usepackage{tikz-cd}
\usepackage{tikz}
\usetikzlibrary{positioning,decorations.markings,arrows.meta,calc,fit,quotes,cd,math,arrows,backgrounds,shapes.geometric}
\usepackage{xy}
\usepackage{amsmath}
\usepackage{listings}
\newcommand{\mlstinline}[1]{\mbox{\lstinline|#1|}}
\xyoption{all}

\newtheorem{desideratum}{Desideratum}

\newcommand{\MainCategory}{\mathbf{ImP}}

\DeclareMathOperator*{\colim}{colim}

\newcommand{\rx}{{\color{red}r}}
\newcommand{\gx}{{\color{ForestGreen}g}}
\newcommand{\bx}{{\color{blue}b}}

\newcommand{\NN}{\mathbb{N}}
\newcommand{\Atoms}{\mathbb{A}}

\begin{document}
\title{Compositional imprecise probability}
\subtitle{A solution from graded monads and Markov categories}
\author{Jack Liell-Cock}	
\author{Sam Staton}		
\begin{abstract} 
\emph{Imprecise probability} is concerned with uncertainty about which probability distributions to use. It has applications in robust statistics and machine learning.

We look at programming language models for imprecise probability. Our desiderata are that we would like our model to support all kinds of composition, categorical and monoidal; in other words, guided by dataflow diagrams. Another equivalent perspective is that we would like a model of synthetic probability in the sense of Markov categories.

Imprecise probability can be modelled in various ways, with the leading monad-based approach using convex sets of probability distributions. This model is not fully compositional because the monad involved is not commutative, meaning it does not have a proper monoidal structure. In this work, we provide a new fully compositional account. The key idea is to name the non-deterministic choices. To manage the renamings and disjointness of names, we use graded monads. We show that the resulting compositional model is maximal and relate it with the earlier monadic approach, proving that we obtain tighter bounds on the uncertainty.

\end{abstract}
\maketitle
\section{Overview}

This paper is about using programming language notations to give compositional descriptions of imprecise probability.
For illustration, consider a situation with three outcomes: red ($\rx$), green ($\gx$) and blue ($\bx$). A precise probability distribution can be understood as a point in the triangle: the corner ($\rx$) represents 100\% certainty of red; the points on the edge between $\gx$ and $\bx$ represent the probability distributions where~$\rx$ is impossible (Figure~\ref{fig:tris}a).
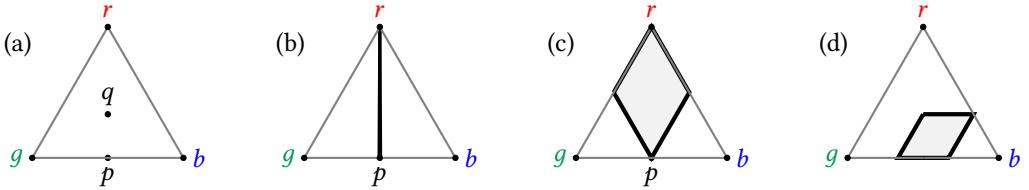
\begin{figure}[h]
  \begin{center}
\begin{tikzpicture}
\filldraw[black] (-1.5,1.5) node[anchor=west]{(a)};
\filldraw[black] (0,0) circle (1pt) node[anchor=north]{$p$};
\filldraw[black] (0,0.577) circle (1pt) node[anchor=south]{$q$};
\draw[gray, thick] (-1,0) -- (1,0);
\draw[gray, thick] (-1,0) -- (0,1.732);
\draw[gray, thick] (1,0) -- (0,1.732);
\filldraw[black] (-1,0) circle (1pt) node[anchor=east]{$\gx$};
\filldraw[black] (1,0) circle (1pt) node[anchor=west]{$\bx$};
\filldraw[black] (0,1.732) circle (1pt) node[anchor=south]{$\rx$};
\end{tikzpicture}
\hspace{.5cm}
  \begin{tikzpicture}
\draw[black,very thick] (-.01,0) -- (0,1.732);
\draw[black,very thick] (.01,0) -- (0,1.732);
  \filldraw[black] (-1.5,1.5) node[anchor=west]{(b)};
\draw[gray, thick] (-1,0) -- (1,0);
\draw[gray, thick] (-1,0) -- (0,1.732);
\draw[gray, thick] (1,0) -- (0,1.732);
\filldraw[black] (0,0) circle (1pt) node[anchor=north]{$p$};
\filldraw[black] (-1,0) circle (1pt) node[anchor=east]{$\gx$};
\filldraw[black] (1,0) circle (1pt) node[anchor=west]{$\bx$};
\filldraw[black] (0,1.732) circle (1pt) node[anchor=south]{$\rx$};
\end{tikzpicture}
\hspace{.5cm}
\begin{tikzpicture}
  \filldraw[color=black, fill=black!5, ultra thick](-.5,.866) -- (0,1.732) -- (.5,.866) -- (0,0) -- (-.5,.866) ;
  \filldraw[black] (-1.5,1.5) node[anchor=west]{(c)};
\filldraw[black] (0,0) circle (1pt) node[anchor=north]{$p$};
\draw[gray, thick] (-1,0) -- (1,0);
\draw[gray, thick] (-1,0) -- (0,1.732);
\draw[gray, thick] (1,0) -- (0,1.732);
\filldraw[black] (-1,0) circle (1pt) node[anchor=east]{$\gx$};
\filldraw[black] (1,0) circle (1pt) node[anchor=west]{$\bx$};
\filldraw[black] (0,1.732) circle (1pt) node[anchor=south]{$\rx$};
\end{tikzpicture}
\hspace{.5cm}
\begin{tikzpicture}\filldraw[color=black, fill=black!5, ultra thick](-1.1102230246251565e-16,0.5773333333333333) -- (0.6666666666666666,0.5773333333333333) -- (0.33333333333333326,0.0) -- (-0.3333333333333334,0.0) -- (-1.1102230246251565e-16,0.5773333333333333);
    \filldraw[black] (-1.5,1.5) node[anchor=west]{(d)};
\draw[gray, thick] (-1,0) -- (1,0); \draw[gray, thick] (-1,0) -- (0,1.732); \draw[gray, thick] (1,0) -- (0,1.732); \filldraw[black] (-1,0) circle (1pt) node[anchor=east]{$\gx$};\filldraw[black] (1,0) circle (1pt) node[anchor=west]{$\bx$};\filldraw[black] (0,1.732) circle (1pt) node[anchor=south]{$\rx$};
\filldraw[black] (0,0) circle (0pt) node[anchor=north]{$\phantom{p}$};
\end{tikzpicture}
\end{center}
\caption{(a)~Five probabilities over the three-point set $\{\rx,\gx,\bx\}$ illustrated as points in the triangle: the three extreme points are the corners; $p$ is the equal odds chance between $\bx$ and $\gx$; $q$ is the equal odds chance between all three points.
  (b)~A line indicating a convex region between $\rx$ and $p$, which includes~$q$.
  (c)~A convex region which is the convex hull of four points, including $\rx$, $p$ and also the equal odds chance between $\rx$ and~$\bx$ and between $\rx$ and $\gx$. (d)~A different convex region, considered in \cite[Ex.~7.3]{walley}.\label{fig:tris}}
\end{figure}

An \emph{imprecise probability} on three outcomes is a convex region of the triangle (Figure~\ref{fig:tris}b--\ref{fig:tris}d). One interpretation is that if a probability distribution describes a bet, as in the foundations of Bayesianism, then a convex region is a collection of bets that would be reasonable given the current imprecise knowledge.
Imprecise probability has a long history in statistical robustness (e.g.~\cite{walley,huber}),
economics (e.g.~\cite{Amarante2007,Battigalli2015,BATTIGALLI2019185,imp-econ,vicigFinancialRiskMeasurement2008,imp-econ2}), and
engineering (e.g.~\cite{fersonConstructingProbabilityBoxes2003,imp-engineering-book,imp-aerospace,imp-intro-book}).
In machine learning, there has been a recent integration of imprecise probability in
Bayesian learning (e.g.~\cite{caprio2024credalbayesiandeeplearning}),
reinforcement learning (e.g.~\cite{oren2024epistemicmontecarlotree,zanger2023}),
conformal prediction (e.g.~\cite{stutz2023conformalpredictionambiguousground,javanmardi2024conformalizedcredalsetpredictors}),
infrabayesianism (e.g.~\cite{infrabayes-physicalism,inframeasure,infradomain}),
and the foundations of safe AI~\cite{aria-1,dalrymple2024guaranteedsafeaiframework}.

There is already a body of work on semantics models of programming languages with imprecise probability~\cite{mio-sarkis-vignudelli,mio-vignudelli,mio,petrisan-sarkis,goy-petrisan,jacobs-2008,adje-static-analysis-imp-prob,gl-previsions,vhos-convex-lang,gl-previsionsB,keimel-plotkin,keimel-mfps,mislove-ouaknine-worrell,varacca-winskel,mciver-morgan}; we discuss this further in \S\ref{sec:notes-about-context}.
Our contribution is to investigate new models that support our compositional desiderata (\S\ref{sec:intro:desiderata}) by naming the non-deterministic choices (\S\ref{sec:intro:names}), and to
investigate syntax for this graded programming language to streamline equational reasoning.
We show that this gives tighter uncertainty bounds than earlier work (Thm.~\ref{thm:unc-bounds}) and that it is a maximal approach (Thm.~\ref{thm:maximality}). 

\subsection{Desiderata: a language for imprecise probability with compositional reasoning}
\label{sec:intro:desiderata}
To focus on the essence of imprecise probability, we look here at a minimal language, that is a first-order functional language without first-class functions or recursion. Rather, we have if-then-else statements, sequencing with immutable variable assignment (like~\cite{moggi:computation_and_monads,freyd-cats}), and the following two commands, which both return a boolean value:
\begin{itemize}
\item \lstinline|bernoulli|: a fair Bernoulli choice~\cite{bernoulli} which draws a ball from some urn containing two balls labelled `true' and `false', and replaces it;
\item \lstinline|knight|: a Knightian choice~\cite{knight} which draws a ball from a fresh urn containing balls labelled `true' and `false', where the number and ratio of balls are unknown and we have no priors on their distribution, except to know that the urn is not empty. (These `Knightian urns' are fresh each time, so they can each be used only once. That is, we are not interested in using multiple draws and frequencies to predict their contents.)
\end{itemize}
For example, consider the following two programs.
\begin{example} The following program, we argue, describes the convex region in Figure~\ref{fig:tris}b:
\label{exA}
\begin{lstlisting}
    x <- knight ; z <- bernoulli ;
    if z then (if x then rx else gx)
         else (if x then rx else bx)
\end{lstlisting}
       We draw two boolean values, \lstinline|x| and \lstinline|z|, respectively with Knightian uncertainty and from a fair Bernoulli trial. We then combine these two boolean values using the logic on the second and third lines of the program.
\end{example}
\begin{example}The following program describes the convex region in Figure~\ref{fig:tris}c:
\label{exB}
\begin{lstlisting}
    x <- knight ; y <- knight ; z <- bernoulli ;
    if z then (if x then rx else gx)
         else (if y then rx else bx)
\end{lstlisting}
This time, we draw three boolean values, \lstinline|x|, \lstinline|y| and \lstinline|z|, where \lstinline|y| is with Knightian uncertainty too. We then combine these three boolean values using the logic on the second and third lines of the program, which is almost the same except for the use of \lstinline|y| when \lstinline|z| is false.
Decoupling the Knightian uncertainties increases the region of imprecise probability because it allows new outcomes (such as an equal chance between $\rx$ and $\bx$ when~\lstinline|x| is true and \lstinline|y| is false) that were impossible in Example~\ref{exA}.
\end{example}

Our desiderata for a \emph{compositional account} of a first-order language are the following. We are inspired by recent compositional accounts of probability theory (e.g.~\cite{kock:commutative-monads-as-a-theory-of-distributions,fritz,jacobs-commutative-effectus}), statistics (e.g.~\cite{partial-markov,jacobs-kissinger-zanasi,braithwaite-compositional-bayes}), and probabilistic programming (e.g.~\cite{lazyppl,DBLP:conf/lics/JiaLMZ21,s-finite}), and the connections between them (e.g.~\cite{stein-staton}). These desiderata are formalized in \S\ref{sec:graded-markov-cats}.

\begin{desideratum} \label{des:1}
The language should be commutative:
\begin{align*}\mlstinline{x <- t ; y <- u ; v} \quad&=\quad \mlstinline{y <- u ; x <- t ; v}
                                                      &&
\text{ (if $x$ is not free in $u$ and $y$ not free in $t$)}
\intertext{and affine:}
                                                         \mlstinline{(x <- t ; u)} \quad&=\quad \mlstinline{u}
                                                                                        &&
                                                                                           \text{ (if $x$ is not free in $u$).}\end{align*}
                                                                                         \end{desideratum}
 This means we can regard composition graphically, as a data flow graph. For instance, the notation
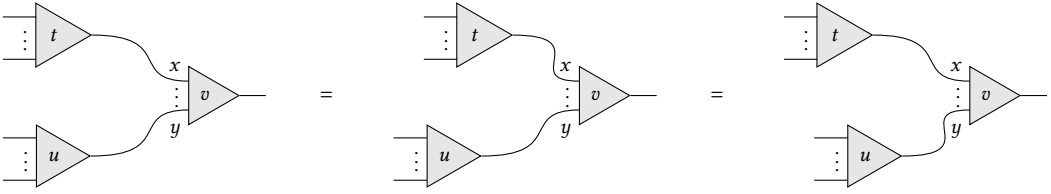
\[
  \begin{tikzpicture}[triangle/.style = {fill=gray!20, regular polygon, regular polygon sides=3}, scale=0.8, every node/.style={transform shape}]
\path (0,0) node [triangle,draw,shape border rotate=-90,label=178:$\vdots$] (a) {$u$} (0,2) node [triangle,draw,shape border rotate=-90,label=178:$\vdots$] (b) {$t$} (2.5,1) node [triangle,draw,shape border rotate=-90,label=135:$x$,label=230:$y$,label=178:$^{\vdots}$] (c) {$v$};
\draw [-] (a) .. controls +(right:2cm) and +(left:1cm).. (c.220);
\draw [-] (b) .. controls +(right:2cm) and +(left:1cm).. (c.140);
\draw [-] (c) to node [above] {} (3.5,1);
\draw [-] (-.85,2.3) to node [above] {} (-.3,2.3);
\draw [-] (-.85,1.6) to node [below] {} (-.3,1.6);
\draw [-] (-.85,0.3) to node [above] {} (-.3,0.3);
\draw [-] (-.85,-0.4) to node [below] {} (-.3,-0.4);
\node () at (4.5,1) {$=$};
\end{tikzpicture}
\qquad
  \begin{tikzpicture}[triangle/.style = {fill=gray!20, regular polygon, regular polygon sides=3},scale=0.8, every node/.style={transform shape}]
\path (0,0) node [triangle,draw,shape border rotate=-90,label=178:$\vdots$] (a) {$u$} (0.5,2) node [triangle,draw,shape border rotate=-90,label=178:$\vdots$] (b) {$t$} (2.5,1) node [triangle,draw,shape border rotate=-90,label=135:$x$,label=230:$y$,label=178:$^{\vdots}$] (c) {$v$};
\draw [-] (a) .. controls +(right:2cm) and +(left:1cm).. (c.220);
\draw [-] (b) .. controls +(right:2cm) and +(left:1cm).. (c.140);
\draw [-] (c) to node [above] {} (3.5,1);
\draw [-] (-0.35,2.3) to node [above] {} (.2,2.3);
\draw [-] (-0.35,1.6) to node [below] {} (.2,1.6);
\draw [-] (-.85,0.3) to node [above] {} (-.3,0.3);
\draw [-] (-.85,-0.4) to node [below] {} (-.3,-0.4);
\node () at (4.5,1) {$=$};
\end{tikzpicture}
\qquad
  \begin{tikzpicture}[triangle/.style = {fill=gray!20, regular polygon, regular polygon sides=3},scale=0.8, every node/.style={transform shape}]
\path (0.5,0) node [triangle,draw,shape border rotate=-90,label=178:$\vdots$] (a) {$u$} (0,2) node [triangle,draw,shape border rotate=-90,label=178:$\vdots$] (b) {$t$} (2.5,1) node [triangle,draw,shape border rotate=-90,label=135:$x$,label=230:$y$,label=178:$^{\vdots}$] (c) {$v$};
\draw [-] (a) .. controls +(right:2cm) and +(left:1cm).. (c.220);
\draw [-] (b) .. controls +(right:2cm) and +(left:1cm).. (c.140);
\draw [-] (c) to node [above] {} (3.5,1);
\draw [-] (-.85,2.3) to node [above] {} (-.3,2.3);
\draw [-] (-.85,1.6) to node [below] {} (-.3,1.6);
\draw [-] (-.35,0.3) to node [above] {} (.2,0.3);
\draw [-] (-.35,-0.4) to node [below] {} (.2,-0.4);
\end{tikzpicture}
\] is not ambiguous. Desideratum~\ref{des:1} is also considered a fundamental aspect of the abstract axiomatization of probability~\cite{kock-comm} with commutativity amounting to Fubini's theorem in the measure-theoretic setting~\cite{kock:commutative-monads-as-a-theory-of-distributions}.

Although this requirement does not hold generally in the presence of memory side effects and mutable variables, we do not have mutable variables here, and it is desirable in a declarative language.
For example, we would like to notate the program from Example~\ref{exA} as
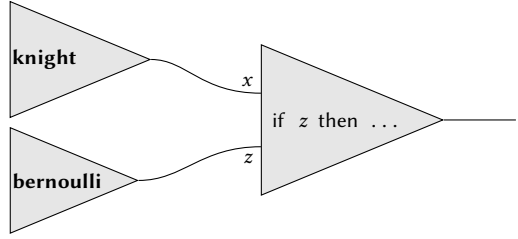
\[
  \begin{tikzpicture}[triangle/.style = {fill=gray!20, isosceles triangle, isosceles triangle stretches}, ,scale=0.8, every node/.style={transform shape}]
\path (0,0) node [triangle,draw,inner sep=1pt] (a) {\lstinline|bernoulli |} (0.053,2) node [triangle, draw,inner sep=1pt] (b) {\lstinline|knight    |} (4.5,1) node [triangle,draw,label=160:$x$,label=200:$z$] (c) {\mlstinline{if z then ...}};
\draw [-] (a) .. controls +(right:2cm) and +(left:1cm).. (c.200);
\draw [-] (b) .. controls +(right:2cm) and +(left:1cm).. (c.160);
\draw [-] (c) to node [above] {} (7.5,1);
\end{tikzpicture}
\]
\begin{desideratum} \label{des:2}
The standard equational reasoning about if-then-else should apply, and in particular, the following hoisting equation should be allowed:
\[
\mlstinline{if b then (x <- t ; u) else (x <- t ; v)}
\quad=\quad
\mlstinline{x <- t ; if b then u else v}
\]
where $x$ is not free in $b$.
\end{desideratum}

This hoisting equation follows from the standard beta and eta laws for if-then-else statements, which are the basic building blocks for equational reasoning.
In models of effects following Moggi's monadic approach~\cite{moggi-computational-lambda, moggi:computation_and_monads}, this property
follows from the base category being distributive, which has long been argued as desirable (e.g.~\cite{cockett-distr}) and thus assumed in all
systems following this approach (e.g.~\cite{dash-thesis, s-finite}). Nevertheless, we briefly discuss the possibility of dropping Desideratum~\ref{des:2} in \S\ref{sec:quotients}.

One earlier approach to a semantic study of a language like this is provided by a convex powerset of distributions monad (e.g.~\cite{keimel-plotkin,mio-sarkis-vignudelli,bonchi-sokolova-vignudelli,bonchi-sokolova-vignudelli-B,mio-vignudelli,goy-petrisan,jacobs-2008}). This does not satisfy the desiderata for compositional reasoning. In fact, no semantic model satisfying the desiderata can allow Examples~\ref{exA} and~\ref{exB} to be distinguished, as we show in Figure~\ref{fig:derivation}.
The key issue is with the third program in Figure~\ref{fig:derivation}:
\begin{lstlisting}
    z <- bernoulli ; if z then (x <- knight ; if x then rx else gx)
                         else (x <- knight ; if x then rx else bx)
\end{lstlisting}
This program draws a boolean value with Knightian uncertainty on each of the branches of the if statement. The paradox arises in whether each choice comes from different urns or the same urn.
Perhaps there is one Knightian draw that is used in both branches. Or perhaps we draw a boolean value from a new Knightian urn on the second branch.
Our proposed solution is to make this distinction explicit.

\begin{figure}
\begin{lstlisting}[frame=single,basicstyle=\small\sffamily]
   x <- knight ; z <- bernoulli ;
   if z then (if x then rx else gx) else (if x then rx else bx)
= -- Desideratum 1 (commutativity)
   z <- bernoulli ; x <- knight ;
   if z then (if x then rx else gx) else (if x then rx else bx)
= -- Desideratum 2
   z <- bernoulli ; if z then (x <- knight ; if x then rx else gx)
                         else (x <- knight ; if x then rx else bx)
= -- Alpha renaming
   z <- bernoulli ; if z then (x <- knight ; if x then rx else gx)
                         else (y <- knight ; if y then rx else bx)
= -- Desideratum 1 (affine)
   z <- bernoulli ; if z then (x <- knight ; y <- knight ; if x then rx else gx)
                         else (x <- knight ; y <- knight ; if y then rx else bx)
= -- Desideratum 2
   z <- bernoulli ; x <- knight ; y <- knight ;
   if z then (if x then rx else gx) else (if y then rx else bx)
= -- Desideratum 1 (commutativity)
   x <- knight ; y <- knight ; z <- bernoulli ; 
   if z then (if x then rx else gx) else (if y then rx else bx)
\end{lstlisting}
\caption{An equational derivation that Examples~\ref{exA} and \ref{exB} must be equal if Desiderata~1 and~2 are satisfied.\label{fig:derivation}}
\end{figure}

\subsection{Resolution: named Knightian choices}
\label{sec:intro:names}
To satisfy both desiderata, \textbf{our proposal} is to name each Knightian choice (\S\ref{sec:our-category}).
To do this, we rewrite Example~\ref{exA} by annotating the only Knightian choice with the name $a_1$:
\begin{lstlisting}
    x <- knight(a1) ; z <- bernoulli ;
    if z then (if x then rx else gx)
         else (if x then rx else bx)
\end{lstlisting}
We think of this program as giving rise to the convex set in Figure~\ref{fig:tris}b.
Now when we try to follow the same equational derivation as in Figure~\ref{fig:derivation}, the third program becomes:
\begin{lstlisting}
    z <- bernoulli ; if z then (x <- knight(a1) ; if x then rx else gx)
                         else (x <- knight(a1) ; if x then rx else bx)
\end{lstlisting}
This program is equivalent even if we were to alpha rename $x$ to $y$ on the second branch because they use the same Knightian choice, i.e.~one with the same name.
\begin{lstlisting}[label=lst:prog1, caption=Convex set in Figure~\ref{fig:tris}b]
    z <- bernoulli ; if z then (x <- knight(a1) ; if x then rx else gx)
                         else (y <- knight(a1) ; if y then rx else bx)
\end{lstlisting}
However, it is different from the program where \lstinline|y| describes a different Knightian choice, i.e.~one with a different name:
\begin{lstlisting}[label=lst:prog2, caption=Convex set in Figure~\ref{fig:tris}c]
    z <- bernoulli ; if z then (x <- knight(a1) ; if x then rx else gx)
                         else (y <- knight(a2) ; if y then rx else bx)
\end{lstlisting}
From this program, we can continue the derivation as in Figure~\ref{fig:derivation}, using affinity and Desideratum~\ref{des:1} to get:
\begin{lstlisting}
    x <- knight(a1) ; y <- knight(a2) ; z <- bernoulli ;
    if z then (if x then rx else gx)
         else (if y then rx else bx)
\end{lstlisting}
which is intuitively what Example~\ref{exB} describes, and gives rise to the convex set in Figure~\ref{fig:tris}c.
Hence, by explicitly labelling the Knightian choices, we have carved a distinction between Examples~\ref{exA} and~\ref{exB}.

The idea of naming non-deterministic choices appears in work outside probability (e.g.~proved transitions \cite{boudol-castellani}) and probabilistic choices are often named in practical probabilistic programming~\cite[\S6.2]{probprog-intro} which has already been explored using graded monads~\cite{lew-trace-types}. More generally, intensionality in non-determinism is known to be a profitable perspective (e.g.~\cite{weighted-relational-model,cfw-bicat-rec-domains}).

\subsubsection{Named Knightian choices via a reader monad}
The set-up with named Knightian choices is consistent with Desiderata~1 and~2, which we can show by building a monad (e.g.~\cite{moggi:computation_and_monads}), namely the reader transformer (e.g.~\cite{monad-trans}) of the finite distributions monad (e.g.~\cite[Ch.~2]{bart-prob-book}):
\begin{equation}\label{eqn:reader}
  T_{2^A}(X) = [2^A\Rightarrow D(X)]
\end{equation}
where $X$ is the set of outcomes, $A$ is the set of names required, and $D$ is the finite distributions monad. Then the Knightian choices are interpreted by reading, and the Bernoulli choices use the distributions monad. This combined monad is well known to be commutative and affine. Thus both desiderata are satisfied.

We can recover a convex set of probability distributions from any $t\in T_{2^A}(X)$ by pushing forward all the possible probability distributions on $2^A$. Formally, we can express this using the monadic bind ($\bind_D$, Kleisli composition) of $D$:
\[
  \llbracket t\rrbracket_{2^A} = \{p\bind_D t~|~p\in D(2^A)\}\subseteq D(X)\text.
\]

\subsubsection{Grading to account for renamings}\label{sec:intro:grading}
A remaining concern with named Knightian choices is that we ought to take seriously name-space issues in composition. When composing programs with named Knightian choices, we must account for name clashes.
This depends on how we interpret the set $A$ in \eqref{eqn:reader}.

We resolve this issue by regarding the monad~\eqref{eqn:reader} as a graded monad~\cite{katsumata-param-monads,orchard-wadler-eades,kammarplotkin}.
This is closely related to the `para' construction (e.g.~\cite{backprop,hermida-tennent}).
The grading is crucial for defining the composition of programs, with the main steps being:
\begin{itemize}
\item Any injection $\iota:A\to B$ induces a renaming of programs using names $A$ to programs using names $B$, and indeed a natural map $T_{2^\iota}:T_{2^A}(X)\to T_{2^B}(X)$;
\item We can regard monadic bind (Kleisli composition) in $T$ as operating on disjoint sets of names:
  \[
    \bind_T\ : \ T_{2^A}(X) \times (X\Rightarrow T_{2^B} (Y)) \to T_{2^{A\uplus B}}(Y)
  \]
  Thus a computation using names $A$ is sequenced with a computation using names~$B$ to build a computation that involves names $(A\uplus B)$.
\item This monad is graded-monoidal too, via a map
  \[T_{2^A}(X)\times T_{2^B}(Y)\to T_{2^{(A\uplus B)}}(X\times Y)
  \]
  which juxtaposes computations using names~$A$ and~$B$ to give a computation using $(A\uplus B)$. 
\item The induced convex set of distributions is invariant under renaming: $\llbracket t\rrbracket_{2^A} = \llbracket T_{2^\iota}(t)\rrbracket_{2^B}$.
\end{itemize}
The key element is that the injective renaming $\iota$ induces a surjection $2^\iota:2^B\to 2^A$
between the spaces of Knightian choices. We abstract and generalize by allowing arbitrary surjections $2^B\to 2^A$, further by allowing sets other than $2^A$, and further still by allowing surjective stochastic maps rather than surjections. 
Quotienting via these regradings moreover leads to the canonical connection of the monad with convex sets of distributions (Proposition~\ref{prop:kan-ext}).

In the rest of the paper, we start from the setting of Markov categories (\S\ref{sec:graded-markov-cats}) rather than monads, since they are a minimal framework suited to our minimal language. We make the connection with monads in \S\ref{sec:markov-monad}, introducing the formalism for named choices in \S\ref{sec:our-category}.

\medskip

As an aside, we note that probability monads too can often be regarded as sort-of reader monads (e.g.~\cite{tavares2021language,lazyppl,DBLP:journals/pacmpl/VakarKS19,barker-monad,Shiebler2021categorical,simpson-giry,de-randvar}), since probability distributions $D(X)$ can be described by random variables $\Omega\to X$, for some base probability sample space $\Omega$. Thus we could regard our monad $T(X)$ as a quotient of 
\[
  [(\Omega\times \Xi)\Rightarrow X]
\]
where $\Omega$ is a sample space for Bernoulli probability and $\Xi=2^A$ is a sample space for Knightian uncertainty. 
In this work, we will quotient by the `law' of random variables in $\Omega$, so that the usual equational reasoning about Bernoulli probability is valid.

\subsubsection{An algebraic perspective}
\newcommand{\probplus}{+_{0.5}}
\newcommand{\ndplus}{\oplus}
\newcommand{\defeq}{\stackrel {\text{def}}=}

There is a large literature on combining probability with non-determinism, some of which take the perspective of algebraic theories. In the context of algebraic effects, the commands
\lstinline|bernoulli| and \lstinline|knight| are \emph{generic effects}~\cite{pp-algop-geneff} for probabilistic and Knightian choice, with which we can define two binary operators:
\[(t\probplus u) \defeq  \mlstinline{if bernoulli then t else u}
  \qquad
  (t\ndplus u) \defeq \mlstinline{if knight then t else u}
\]
We provide a more detailed discussion on this perspective in \S\ref{sec:alg-perspective} and \S\ref{sec:graded-alg}.
In brief, Desideratum~\ref{des:1} states that these operators commute with each other~\eqref{eqn:commutativity} and are idempotent~\eqref{eqn:affine}. Desideratum~\ref{des:2} is always assumed in algebraic effects as a consequence of the algebraicity property.
Symmetry of the operators is also desirable but is incompatible with the Desiderata~\eqref{eqn:eckmann-hilton}.
Our graded approach to Knightian choice (\S\ref{sec:intro:grading}) connects to graded algebraic theories~\cite[\S3.1]{kura-graded-algebraic} which amounts to naming the non-deterministic binary operators and allows for symmetry of the operators up to regrading~\eqref{eqn:assocgraded}.

\subsection{Results about quotienting our theory}
\newcommand{\CP}{\mathrm{CP}}
\newcommand{\CPf}{\CP}
\label{sec:intro:results}
The names for the Knightian choices in our language appear to be additional intensional information, and the reader monad does not quotient this away. For this reason, we show two results about the equational theory.
First, we connect our approach to the convex powerset of distributions monad, showing that our bounds are tighter. Second, we show it is maximal --- no further quotient is possible.

\paragraph{Theorem~\ref{thm:unc-bounds} (\S\ref{sec:oplax-functor}): Improved bounds on uncertainty}
In our resulting language, every closed term describes a convex set of distributions. We thus establish a connection to the non-compositional approach that uses the Kleisli category of the convex powerset of distributions monad (e.g.~\cite{keimel-plotkin,mio-sarkis-vignudelli,bonchi-sokolova-vignudelli,bonchi-sokolova-vignudelli-B,mio-vignudelli,goy-petrisan,jacobs-2008}). We have an `op-lax' functor
\[
R:  \MainCategory \to \mathrm{Kl}(\CP).
\]
from our locally graded category $\MainCategory$ (\S\ref{sec:maincategory}) to the Kleisli category of the convex powerset of distributions monad~$\CP$. 
Being an op-lax functor means that
\[R(g\circ f)\subseteq R(g)\circ R(f),\]
i.e.~composition in our category gives a tighter bound on the Knightian uncertainty than the composition using the Kleisli category of the convex powerset of distributions monad (see also the Example in \S\ref{sec:oplax-example}). 

Note that this could not be a proper functor because we would then have a quotient theory in violation of the maximality theorem (Theorem~\ref{thm:maximality}). But an op-lax functor is beneficial as an interpretation of giving a tighter bound. 

\paragraph{Theorem~\ref{thm:maximality} (\S\ref{sec:max}): Maximality.}
Our language also gives rise to a compositional theory of equality. We prove our equational theory is maximal in that we can add no further equations on open terms without equating different convex sets of distributions or compromising the compositional structure. (See Theorem~\ref{thm:maximality} for a precise statement.)

\section{Markov categories and programming language syntax}
\label{sec:graded-markov-cats}
Markov categories have been proposed as a synthetic foundation of probability~\cite{fritz}, but they also form a basic framework for equational reasoning about probabilistic programming (e.g.~\cite{stein-staton,afkkmrsy}). As we show, distributive Markov categories precisely capture our Desiderata (\S\ref{sec:intro:desiderata}).
Markov categories are connected to monad-based approaches too, as we discuss in \S\ref{sec:markov-monad}.

We begin by recalling notions of Markov categories and how we may interpret programming language syntax in them (\S\ref{sec:markov-def}), and then generalize them
to the locally graded setting (\S\ref{sec:graded-markov}).
\newcommand{\Kl}{\mathrm{Kl}}
\newcommand{\CatA}{\mathcal{C}}
\newcommand{\CCat}{\mathbf{C}}
\newcommand{\DCat}{\mathbf{D}}
\newcommand{\VCat}{\mathcal{V}}
\newcommand{\GCat}{\mathbb{G}}
\newcommand{\JCat}{\mathcal{J}}
\newcommand{\Set}{\mathbf{Set}}
\newcommand{\FinStoch}{\mathbf{FinStoch}}
\newcommand{\FinSet}{\mathbf{FinSet}}
\newcommand{\FinStochSurj}{\FinStoch_{\mathrm{Surj}}}
\newcommand{\FinSetSurj}{\FinSet_{\mathrm{Surj}}}
\newcommand{\inv}{^{-1}}
\newcommand{\op}{^{\mathrm{op}}}
\begin{definition}
  A \emph{monoidal category} is a category $\CatA$ equipped with
  a functor $\otimes :\CatA\times \CatA\to \CatA$ and an object $I$
  together with associativity and unitor isomorphisms (e.g. $(X\otimes Y)\otimes Z\cong X\otimes (Y\otimes Z)$) that satisfy coherence conditions (e.g.~\cite{maclane}). It is strict if the isomorphisms are equalities. 
  A \emph{symmetric} monoidal category is moreover equipped with isomorphisms
  $\sigma_{X,Y}:X\otimes Y\cong Y\otimes X$ such that
  $\sigma_{Y,X}=\sigma_{X,Y}\inv$ and satisfying coherence conditions. 

  A \emph{semicartesian category} is a symmetric monoidal category in which the monoidal unit is a terminal object. That is, there is exactly one morphism $X\to I$ for all $X$. 
\end{definition}
A semicartesian category has projections $X\otimes Y\to X\otimes I\cong X$, but it is weaker than a full categorical product because there need not be a natural diagonal $X\to X\otimes X$. 
\subsection{Markov categories}
\label{sec:markov-def}
\newcommand{\copymap}{\mathsf{copy}}
\newcommand{\id}{\mathsf{id}}
\begin{definition}\cite{fritz}
  A \emph{Markov category} is a semicartesian category such that every object is equipped with a commutative comonoid structure, that is, a map
  $\copymap_X: X\to X\otimes X$ that is symmetric and associative and has the terminal map
  $X\to I$ as a unit.

  A morphism $f: X\to Y$ in a Markov category is \emph{deterministic} if it commutes with the copy map ($\copymap_Y\circ f=(f\otimes f)\circ \copymap_X$).

  A \emph{distributive} Markov category~\cite{afkkmrsy} is a Markov category that has coproducts such that the canonical maps $X\otimes Z+Y\otimes Z\to (X+Y)\otimes Z$ are isomorphisms and the coproduct injections $X\to X+Y\leftarrow Y$ are deterministic.\end{definition}
A typical example of a distributive Markov category is the category $\FinStoch$ of stochastic matrices (Def.~\ref{def:finstoch}).

An ordinary distributive category~\cite{cockett-distr,clw} is a distributive Markov category in which every morphism is deterministic. A typical example is the category $\FinSet$ of finite sets.

\paragraph{Programming syntax.} We can use programming language syntax for composition in a distributive Markov category (see also e.g.~\cite{stein-staton}). The objects of the category are
regarded as types, with \lstinline|Bool| regarded as the object $1+1$. If $\Gamma=(x_1:A_1)\otimes \dots \otimes (x_n:A_n)$ then a morphism $t:\Gamma \to B$ is regarded as a term
$\Gamma \vdash t :B$. We notate
\begin{align}
  &\mlstinline{(t,u)}&\text{for}&&& \Gamma\xrightarrow\copymap \Gamma\otimes \Gamma \xrightarrow {t\otimes u}
                                    A\otimes B
  \\
  &\mlstinline{x <- t ; u} &\text{for}&&& \Gamma \xrightarrow \copymap \Gamma\otimes \Gamma \xrightarrow{\Gamma \otimes t} \Gamma \otimes A\xrightarrow u B \label{eq:markov-let}
  \\
  &\mlstinline{if t then u else v}&\text{for}&&& \Gamma \xrightarrow \copymap \Gamma\otimes \Gamma
                                             \xrightarrow {t\otimes \Gamma} (1+1)\otimes \Gamma \cong \Gamma+\Gamma
                                             \xrightarrow {[u,v]} B
\end{align}
In this way, given interpretations of \lstinline|bernoulli| and \lstinline|knight| as morphisms
$1\to \mlstinline{Bool}$, we can interpret the programs from Examples~\ref{exA} and~\ref{exB}. 

\subsection{Graded Markov categories} \label{sec:graded-markov}
\begin{definition}\label{def:graded-markov}
  Let $\GCat$ be a semicartesian category. A \emph{graded distributive Markov category} $\CCat$ is given by
  \begin{itemize}
    \item a distributive Markov category $\CCat_I$, but moreover,
    \item for each pair of objects and each grade $\gamma\in \GCat$
      a set $\CCat_\gamma(X,Y)$ of morphisms, agreeing with $\CCat_I$ when $\gamma=I$;
    \item for each morphism $f:\varepsilon\to \gamma$ in $\GCat$, a function
      $\CCat_\gamma(X,Y)\to \CCat_\varepsilon(X,Y)$;
    \item a family of maps
      $\circ: \CCat_\gamma(X,Y)\times \CCat_\varepsilon(Y,Z)\to \CCat_{\gamma\otimes \varepsilon}(X,Z)$;
    \item a family
      $\otimes: \CCat_\gamma(X,X')\times \CCat_\varepsilon(Y,Y')\to
      \CCat_{\gamma\otimes \varepsilon}(X\otimes Y,X'\otimes Y')$
  \end{itemize}
  all such that composition is natural and associative up to the associativity of $\GCat$ (see e.g.~\cite[\S1.2]{wood-graded},~\cite{pbl-graded}, \cite[App.~B]{gavranovic-thesis}), monoidal product of morphisms is also natural and has associators and symmetric braidings up-to the structure of $\GCat$, and such that the induced function
    $\CCat_\gamma(X+Y,Z)\to \CCat_{\gamma}(X,Z)\times \CCat_{\gamma}(Y,Z)$
  is a bijection (e.g.~\cite[p.~36]{wood-graded}).
\end{definition}
See Proposition~\ref{prop:paragraded} for our example of a graded Markov category. We note that since $\GCat$ is semicartesian, there are canonical projections $\gamma\otimes \varepsilon\to \gamma$,
and so we can regard any morphism at grade $\gamma$ as a morphism at grade $(\gamma\otimes \varepsilon)$.

\newcommand\atg{\mathrel{{\&}}}
\subsection{Graded programming syntax}\label{sec:graded-pl} We propose to use programming language syntax for composition in a graded distributive Markov category, as in a distributive Markov category.
Our syntax and interpretation are similar to other languages for grades and effect systems, e.g.~\cite[\S2.3]{mcdermott-thesis}.
The objects of the category are again regarded as types
with \lstinline{Bool} and contexts $\Gamma$ regarded as before. Then a morphism at grade $\gamma$, e.g.~$t\in\CCat_\gamma(\Gamma ,B)$, is regarded as a term at grade $\gamma$, written~$\Gamma\vdash t: B \atg \gamma$.
We can build a simple internal language for the graded category using typing rules such as the following:
\newcommand{\regrade}[2]{\mathsf{c}_{#1}[#2]}
\[\begin{array}{l@{\qquad}l}
  \begin{prooftree}
\phantom{\Gamma \vdash u:A \atg g}   \justifies
    \Gamma, x:A \vdash x:A \atg I
  \end{prooftree}
  &
  \begin{prooftree}
    \Gamma \vdash u:A \atg \gamma
    \qquad
    \Gamma, x:A \vdash t:B \atg \varepsilon
    \justifies
    \Gamma \vdash \mlstinline{x <- u; t} :B \atg \gamma\otimes \varepsilon
  \end{prooftree}
    \\[24pt]
  \begin{prooftree}
    \Gamma \vdash t:A \atg \gamma  \justifies
    \Gamma\vdash \regrade f t:A \atg \varepsilon
    \using \mbox{\small $f:\varepsilon\to \gamma$}
  \end{prooftree}
    &\begin{prooftree}
    \Gamma \vdash b:\mlstinline{Bool} \atg \gamma
    \qquad
    \Gamma \vdash u:A \atg \varepsilon
    \qquad
    \Gamma \vdash t:A \atg \varepsilon
    \justifies
    \Gamma \vdash \mlstinline{if b then u else t} : A \atg \gamma\otimes \varepsilon
  \end{prooftree}\end{array}
  \hide{\qquad
  \begin{prooftree}
    \Gamma \vdash u:A \atg \gamma
    \qquad
    \Gamma \vdash t:B \atg \varepsilon
    \justifies
    \Gamma \vdash (u,t): A\otimes B \atg \gamma\otimes \varepsilon
  \end{prooftree}}
\]
The notation $\regrade f t$ is a regrading or coercion, corresponding to the functorial action on hom-sets
$\CCat_{\gamma}(\Gamma, A)\to \CCat_{\varepsilon}(\Gamma, A)$ in the graded Markov category. 

There are also typing rules for product types and sums more generally, but since we won't need them
in examples, we omit them here.
We do use a three-element type in the illustrations, $\mlstinline{Three}=1+1+1$, with constructors
\[
    \Gamma \vdash \rx: \mlstinline{Three} \atg I
    \qquad
    \Gamma \vdash \gx: \mlstinline{Three} \atg I
    \qquad
    \Gamma \vdash \bx: \mlstinline{Three} \atg I
  \]
For the language of imprecise probability with named Knightian choices, we suppose that there
is a grade $\widehat a\in \GCat$ that represents the grading by a single name $a$.
We then consider the following commands:
\[
    \Gamma \vdash \mlstinline{bernoulli} :\mlstinline{Bool} \atg I
  \qquad
    \Gamma \vdash \mlstinline{knight(a)} :\mlstinline{Bool} \atg \widehat{a}
\]
These will be interpreted via specific morphisms
in $\CCat_I(I,\mlstinline{Bool})$ and $\CCat_{\widehat a}(I,\mlstinline{Bool})$ respectively. 

\subsubsection{Example and regrading by coherence isomorphisms}
We then have the following derivable type, for a program from \S\ref{sec:intro:names}:
\begin{align*}
  \Gamma\vdash&\ 
                \mlstinline{x <- knight(a1) ; y <- knight(a2) ; z <- bernoulli ;}
  \\&\ \mlstinline{if z then (if x then rx else gx) else (if y then rx else bx)}
\   : \mlstinline{Three}\atg (\widehat{a_1}\otimes \widehat{a_2}\otimes I\otimes I\otimes I)
\end{align*}
We can regrade to $\widehat{a_1}\otimes \widehat{a_2}$ via the
coherence isomorphism $(\widehat{a_1}\otimes \widehat{a_2}) \cong (\widehat{a_1}\otimes \widehat{a_2}\otimes I\otimes I\otimes I)$ in $\GCat$.

In some situations, $\GCat$ can be chosen so that these coherence isomorphisms are identity maps so this trivial regrading can be omitted. 
For example, in many situations in the literature, $\GCat$ is a partially ordered monoid.
In \S\ref{sec:choices-products} we suggest choices for $\GCat$ that avoid coherence regradings. 

\subsubsection{Equational reasoning}
\hide{\[
  \begin{prooftree}
    \Gamma \vdash b \equiv b':\mlstinline{bool} \atg \gamma
    \qquad
    \Gamma \vdash u \equiv u':A \atg \varepsilon
    \qquad
    \Gamma \vdash t\equiv t':A \atg \varepsilon
    \justifies
    \Gamma \vdash \mlstinline{if b then u else t} \equiv \mlstinline{if b' then u' else t'} : A \atg \gamma\otimes \varepsilon
  \end{prooftree}
\]

\[
  \begin{prooftree}
    \Gamma \vdash u \equiv u':A \atg \gamma
    \qquad
    \Gamma, x:A \vdash t\equiv t':B \atg \varepsilon
    \justifies
    \Gamma \vdash \mlstinline{x <- u; t} \equiv \mlstinline{x <- u'; t'} :B \atg \gamma\otimes \varepsilon
  \end{prooftree}
  \qquad
  \begin{prooftree}
    \Gamma \vdash u \equiv u':A \atg \gamma
    \qquad
    \Gamma \vdash t \equiv t':B \atg \varepsilon
    \justifies
    \Gamma \vdash (u,t) = (u',t'): A\otimes B \atg \gamma\otimes \varepsilon
  \end{prooftree}
\]}
We have the following equational reasoning principles. The first is associativity of composition,
familiar from the monadic metalanguage~\cite{moggi-computational-lambda,moggi:computation_and_monads}.
The second and third are commutativity and affine laws (Desiderata~1), and finally the hoisting law (Desiderata~2).
Since we write `$\Gamma\vdash t:A\atg \gamma$' for a morphism $t\in\CCat_\gamma(\Gamma, A)$,
we also write `$\Gamma\vdash t\equiv u:A\atg \gamma$' to indicate that the morphisms $t$ and $u$ are equal. 
\begin{proposition}\label{eqn:graded-pl-sound}In any graded distributive Markov category, the following equational reasoning principles are valid:
  \begin{enumerate}
  \item (Associativity)\\
    \begin{tabular}{@{}ll}If &$\Gamma \vdash t :A \atg \gamma
    \qquad
    \Gamma, x:A \vdash u: B \atg \varepsilon
    \qquad
    \Gamma, y:B \vdash v: C \atg \zeta$
\\then &$
    \Gamma \vdash \mlstinline{x <- t ; y <- u ; v} \ \  \equiv\ \ \regrade{\alpha}{\mlstinline{y <- (x <- t ; u) ; v}} : C \atg \gamma\otimes (\varepsilon\otimes \zeta)$\end{tabular}\\
  where $\alpha: \gamma\otimes (\varepsilon\otimes \zeta)\to (\gamma\otimes \varepsilon)\otimes \zeta$ is the associativity coherence isomorphism in $\GCat$. 
  \item (Commutativity)\\
    \begin{tabular}{@{}ll}
      If &$   \Gamma \vdash t :A \atg \gamma
    \qquad
    \Gamma \vdash u: B \atg \varepsilon
    \qquad
    \Gamma,x:A,y:B \vdash u: C \atg \zeta
$
  \\then &$
  \Gamma \vdash \regrade {\alpha\inv} {\mlstinline{x <- t; y <- u; v}} \ $\\
      &$\phantom{\Gamma}\ \equiv\ 
  \regrade  {(\sigma\otimes \zeta)\circ \alpha\inv}{\mlstinline{y <- u; x <- t; v}} : C \atg (\gamma\otimes \varepsilon)\otimes \zeta
$\end{tabular}\\
where $\sigma:\gamma\otimes \varepsilon\to \varepsilon\otimes \gamma$ is the symmetry of the monoidal structure in $\GCat$.
  \item (Weakening)\\\begin{tabular}{@{}ll}
    If &$    \Gamma \vdash t :A \atg \gamma
    \qquad
    \Gamma \vdash u: B \atg \varepsilon
  $\\
  then &
  $
    \Gamma \vdash \mlstinline{x <- t ; u} \ \equiv\  \regrade {\pi_2} u : B \atg \gamma\otimes \varepsilon
$\end{tabular}\\
where $x\not\in\Gamma$, and $\pi_2:\gamma\otimes \varepsilon \to \varepsilon$ is the canonical projection map from the semicartesian structure of $\GCat$.
  \item (Hoisting)
    \\\begin{tabular}{@{}ll}
    If &$    \Gamma \vdash b :\mlstinline{Bool} \atg \gamma
    \qquad
    \Gamma \vdash t :A \atg \varepsilon
    \qquad
    \Gamma,x:A \vdash u: B \atg \zeta
    \qquad
    \Gamma,x:A \vdash v: B \atg \zeta
  $\\
  then &
  $
  \Gamma \vdash \regrade {\alpha\inv}{\mlstinline{if b then x <- t ; u else x <- t ; v}} \ $\\
        &$\phantom{\Gamma}\ \equiv\ 
    \regrade {(\sigma\otimes \zeta)\circ \alpha\inv} {\mlstinline{x <- t ; if b then u else v}} : B \atg (\gamma\otimes \varepsilon)\otimes \zeta
$\end{tabular}\\
\end{enumerate}
\end{proposition}
\begin{proof}[Proof note]
  By expanding the definitions of the syntactic notation in terms of composition in graded distributive Markov categories.
\end{proof}
When the coherence isomorphisms $c_f$ are identity morphisms, the regradings can be omitted.
We give candidates for $\GCat$ that allow for this simplification in \S\ref{sec:choices-products}.
The weakening regrading ($\pi_2$) is also typically obvious and so we sometimes elide it. 
We leave a practical coercion inference algorithm to future work. 

Figure~\ref{fig:derivation2} illustrates a program equational derivation in this setting, which is similar to the one in Figure~\ref{fig:derivation}, but now formally grounded.

These are not the only laws: there are also equational laws such as reflexivity, symmetry, transitivity, congruence, unit laws for the sequencing construction, and beta and eta laws for sum and product types.  As we will see in \S\ref{sec:markov-monad},
graded Markov categories can be understood in terms of enriched monads, and so we have the full power of an enriched version of the monadic metalanguage~\cite{moggi-computational-lambda,moggi:computation_and_monads}. There are no nuances involved in specifying the other laws, beyond the regrading coercions, so for the sake of brevity, we do not write them all out here.

\begin{figure}
\begin{lstlisting}[frame=single,basicstyle=\small\sffamily]
   z <- bernoulli ; if z then (x <- knight(a1) ; if x then rx else gx)
                         else (x <- knight(a2) ; if x then rx else bx)
= -- Alpha renaming
   z <- bernoulli ; if z then (x <- knight(a1) ; if x then rx else gx)
                         else (y <- knight(a2) ; if y then rx else bx)
= -- Weakening
   z <- bernoulli ; if z then (x <- knight(a1) ; y <- knight(a2) ; if x then rx else gx)
                         else (x <- knight(a1) ; y <- knight(a2) ; if y then rx else bx)
= -- Hoisting
   z <- bernoulli ; x <- knight(a1) ; y <- knight(a2) ;
   if z then (if x then rx else gx) else (if y then rx else bx)
= -- Commutativity
   x <- knight(a1) ; y <- knight(a2) ; z <- bernoulli ;
   if z then (if x then rx else gx) else (if y then rx else bx)
\end{lstlisting}
\caption{A formal equational derivation of two programs with named Knightian choices from \S\ref{sec:intro:names} using the principles of Proposition~\ref{eqn:graded-pl-sound}.\label{fig:derivation2}}
\end{figure}

\section{Relating graded Markov categories and graded monads}
\label{sec:markov-monad}

In \S\ref{sec:graded-markov-cats}, we looked at categorical probability and probabilistic programming from the point of view of Markov categories. Monads provide a different view on this. Because this angle will be more familiar to some in the community, we now recall graded and relative monads (\S\ref{sec:monad-def}) and illustrate their correspondence with Markov categories (Prop.~\ref{prop:monad-markov}, \ref{prop:markov-monad}) and we relate them to notions from enriched
category theory (\S\ref{sec:enriched-def}). For brevity, in this section, we focus on definitions and in the next section (\S\ref{sec:our-category}) we focus on examples,
rather than interleaving them.

\subsection{Monads and graded relative affine monads}
\label{sec:monad-def}

It is well-established that notions of computation can be modelled by monads~\cite{moggi:computation_and_monads}, including probabilistic and non-deterministic computation (already in~\cite{moggi-computational-lambda,jones_probabilistic_1989}).
In this section, we recall the flavours of monads relevant to this work.

The emerging view is that \emph{commutative affine monads} form abstract accounts of probability and non-determinism (e.g.~\cite{jacobs-commutative-effectus,bart-prob-book,kock:commutative-monads-as-a-theory-of-distributions,lazyppl}). Meanwhile, \emph{relative monads} restrict the domain of first-order computation (e.g.~\cite{asada-arrows}),
and \emph{graded monads} classify the side-effects associated with a program (e.g.~\cite{orchard-wadler-eades,katsumata-param-monads}).
We present these concepts in the Kleisli triple setting because this is more conducive to their use in programming languages.
\begin{definition}
  A \emph{strong monad}~\cite{kockStrongFunctorsMonoidal1972,moggi:computation_and_monads} over a cartesian closed category $\CCat$ is for each $X\in\CCat$ an object $T(X)\in\CCat$ and a morphism
  \[
    \eta_X : 1_\CCat \to [X, T(X)],
  \]
  and a family of morphisms
  \[
    (-)^* : {[X, T(Y)]} \to {[T(X), T(Y)]}
  \] in $\CCat$ such that for generalised elements $f$ and $g$ of $[X, T(Y)]$ and $[Y, T(Z)]$, the following equations hold:
  \begin{equation}
    \begin{aligned} \label{eq:monad_laws}
      f &= f^*\circ \eta_X \\
      id_{T(X)} &= (\eta_X)^* \\
      g^* \circ f^* &= (g^* \circ f)^*.
    \end{aligned}
  \end{equation}
  The left-strength $s: X\times T(Y) \to T(X\times Y)$ is induced by the canonical action $(\eta_Y \circ {-})^* : [X,Y] \to [T(X), T(Y)]$ by
  \[
    X\times T(Y) \to [Y, X\times Y]\times T(Y) \to [T(Y), T(X\times Y)]\times T(Y) \to T(X\times Y),
  \]
  where the first and last arrows use the unit and the counit of the closed structure, respectively.

  The strong monad is \emph{commutative}  if the following diagram commutes, where $\hat{s}$ is the induced right-strength from the symmetry of the cartesian product.
  \[
  \begin{tikzpicture}[commutative diagrams/every diagram]
    \node (P0) at (90:1.7cm) {$T(X)\times T(Y)$};
    \node (P1) at (90+90:2cm) {$T(X\times T(Y))$} ;
    \node (P2) at (90+2*90:1.7cm) {$T(X\times Y)$};
    \node (P3) at (90+3*90:2cm) {$T(T(X)\times Y)$};
    \path[commutative diagrams/.cd, every arrow, every label]
    (P0) edge node[swap] {$\hat{s}$} (P1)
    (P1) edge node[swap] {$s^*$} (P2)
    (P3) edge node {$\hat{s}^*$} (P2)
    (P0) edge node {$s$} (P3);
  \end{tikzpicture}
  \]
  The strong monad is \emph{affine} (e.g.~\cite{kock-affine,jacobs-weakening} if the unique map $T(1)\to 1$ is an isomorphism. A typical example is $\CCat=\Set$, and $T=D$ is the finite probability distribution monad (e.g.~\cite{bart-prob-book}). 

  Let $\JCat$ be a category with finite products and consider a finite product preserving functor $J:\JCat\to \CCat$.
  A \emph{relative strong monad}~\cite{relative-monads,uustalu-strong-relative} $T$ on $J$ is a functor $T: \JCat \to \CCat$, along with a $J$-relative unit
  \[
    \eta_X: I \to [J(X), T(X)]
  \]
  natural in $X\in\JCat$, and a family of $J$-relative Kleisli extensions
  \[
    ({-})^*: [J(X),T(Y)] \to [T(X), T(Y)]
  \]
  natural in $X,Y\in \JCat$, and such that \eqref{eq:monad_laws} holds for $f$ and $g$ generalised elements of $[J(X), T(Y)]$ and $[J(Y), T(Z)]$.
  A typical example is $\CCat=\Set$, and $\JCat=\FinSet$, with $J$ the evident embedding, and $T=DJ$. 

  Let $(\GCat, \otimes, I)$ be a monoidal category.
  A \emph{graded strong monad} (e.g.~\cite{katsumata-param-monads}) is a functor $T: \GCat \to [\CCat,\CCat]$, with  unit
  \[
    \eta_X: I \to [X, T_I(X)]
  \]
  natural in $X\in\CCat$, and a family of Kleisli extensions
  \[
    ({-})^*_{\gamma,\varepsilon}: [X,T_\varepsilon(Y)] \to [T_{\gamma}(X), T_{\gamma\otimes \varepsilon}(Y)]
  \]
  natural in $X,Y\in \CCat$, and such that for $f$ and $g$ generalised elements of $[X, T_\gamma(Y)]$ and $[Y, T_\varepsilon(Z)]$, the following equations hold:
  \begin{equation} \label{eq:graded_monad_laws}
    \begin{aligned}
      f &= T_\lambda \circ (f)^*_{I,\gamma} \circ \eta_{X}, \\
      id_{T_\varepsilon(X)} &= T_\rho \circ (\eta_X)^*_{\varepsilon,I}, \\
      (g)^*_{\zeta\otimes \gamma, \varepsilon}\circ (f)^*_{\zeta, \gamma} &= T_\alpha \circ ((g)^*_{\gamma, \varepsilon}\circ f)^*_{\zeta, \gamma\otimes \varepsilon},
    \end{aligned}
  \end{equation}
  where $\lambda: I\otimes \gamma \to \gamma$, $\rho: \varepsilon\otimes I \to \varepsilon$, and $\alpha: \zeta\otimes (\gamma\otimes \varepsilon) \to (\zeta\otimes \gamma) \otimes \varepsilon$
  are the left unitor, right unitor, and associator of $\GCat$, respectively.

  A \emph{graded relative strong monad} $T$ on $J$ is a functor $T: \GCat \to [\JCat,\CCat]$, along with a $J$-relative unit
  \[
    \eta_X: I_\CCat \to [J(X), T_I(X)]
  \]
  natural in $X\in\JCat$, and a family of $J$-relative Kleisli extensions
  \[
    ({-})^*_{\gamma,\varepsilon}: [J(X),T_\varepsilon(Y)] \to [T_{\gamma}(X), T_{\gamma\otimes \varepsilon}(Y)]
  \]
  natural in $X,Y\in \JCat$, and such that \eqref{eq:graded_monad_laws} holds for $f$ and $g$ generalised elements of $[J(X), T_\gamma(Y)]$ and $[J(Y), T_\varepsilon(Z)]$.
  The graded left-strength $s_\gamma: X\times T_\gamma(Y) \to T_\gamma(X\times Y)$ is similarly induced by the action $T_\rho \circ (\eta_Y \circ {-})^*_{\gamma,I} : [X,Y] \to [T_\gamma(X), T_\gamma(Y)]$.
  The monad is \emph{commutative} if $\GCat$ is symmetric monoidal and the following diagram commutes,
  where $\sigma: \gamma\otimes \varepsilon \to \varepsilon\otimes \gamma$ is the symmetric coherence isomorphism of $\GCat$.
  \[
  \begin{tikzpicture}[commutative diagrams/every diagram]
    \node (P0) at (90:2.3cm) {$T_\gamma(X)\times T_\varepsilon(Y)$};
    \node (P1) at (90+72+4:2.3cm) {$T_\gamma(X\times T_\varepsilon(Y))$} ;
    \node (P2) at (90+2*72:2cm) {\makebox[5ex][r]{$T_{\gamma\otimes \varepsilon}(X\times Y)$}};
    \node (P3) at (90+3*72:2cm) {\makebox[5ex][l]{$T_{\varepsilon\otimes \gamma}(X\times Y)$}};
    \node (P4) at (90+4*72-4:2.3cm) {$T_\varepsilon(T_\gamma(X)\times Y)$};
    \path[commutative diagrams/.cd, every arrow, every label]
    (P0) edge node[swap] {$\hat{s}_\gamma$} (P1)
    (P1) edge node[swap] {$(s_\varepsilon)^*_{\gamma,\varepsilon}$} (P2)
    (P2) edge node {$T_\sigma$} (P3)
    (P4) edge node {$(\hat{s}_\gamma)^*_{\varepsilon,\gamma}$} (P3)
    (P0) edge node {$s_\varepsilon$} (P4);
  \end{tikzpicture}
  \]
  It is \emph{affine} if the unique map $T_I(1)\to 1$ is an isomorphism.
\end{definition}
\begin{proposition}\label{prop:monad-markov}
  If $T$ is a graded commutative affine relative monad on a distributive category, then its Kleisli category (e.g.~\cite{kura-graded-algebraic,gkos-graded-hoare}) is a graded Markov category:
  \begin{itemize}
  \item The objects are the same as $\JCat$;
  \item The morphisms in $\Kl(T)_\gamma(X,Y)$ are the morphisms $J(X)\to T_\gamma(X)$ in $\CCat$;
  \item Composition is via the Kleisli extension.\end{itemize}
\end{proposition}
\begin{proposition}\label{prop:markov-monad}
  Any graded Markov category induces a graded commutative affine relative monad, by
  \begin{itemize}
  \item $\mathcal J=\CCat_{I,\mathrm{det}}$, the distributive category of $I$-graded deterministic maps
  \item The underlying category is $\mathrm{FP}(\mathcal J\op,\Set)$, the finite product-preserving
    contravariant presheaves on $\mathcal {J}$; 
  \item $J: \mathcal J\to \mathrm{FP}(\mathcal J\op,\Set)$ is the Yoneda embedding;
  \item $T: \GCat\to [\mathcal J,\mathrm{FP}(\mathcal J\op,\Set)]$ is given by
    \[T_\gamma(Y)(X)=\CCat_\gamma(X,Y)\]
  \end{itemize}
\end{proposition}
\begin{proof}[Proof note] In both cases, the proof amounts to expanding the definitions. The constructions are similar to \cite[\S7]{power-universal}. See also \cite[Prop.~13]{afkkmrsy} for the non-graded case.\end{proof}
(We conjecture that Propositions~\ref{prop:monad-markov}--\ref{prop:markov-monad} are part of a biequivalence between graded distributive Markov categories and commutative affine graded relative monads.
We do not pursue this here because we will not need the generality of the biequivalence in what follows.)
\subsection{Connection with enriched categories}
\label{sec:enriched-def}
To show that the concepts in this section are canonical, we connect with the theory of enriched categories.
Let $\VCat$ be a symmetric monoidal closed category with limits and colimits, that is moreover semicartesian. 
Recall (e.g.~\cite{kelly}) that a $\VCat$-enriched category $\CCat$ is given by a collection of objects,
and for each pair of objects $X,Y$ of $\CCat$, a `hom-object' $\CCat(X,Y)$ in $\VCat$.
Composition is a morphism $\CCat(X,Y)\otimes \CCat(Y,Z)\to \CCat(X,Z)$ in $\VCat$.
We can also define $\VCat$-enriched monoidal categories, by requiring the functor $\otimes:\CCat\times \CCat\to \CCat$ to be $\VCat$-enriched.
And $\VCat$-enriched coproducts require a natural isomorphism
\[
  \CCat(X_1+\dots+X_n,Y)\cong   \CCat(X_1,Y)\times \dots \times \CCat(X_n,Y)
\]
between objects of $\VCat$.
Any enriched category has an underlying \emph{ordinary category} $\CCat_0$, which has the same objects but with a hom-\emph{set} given by $\CCat_0(X, Y)=\VCat(I,\CCat(X, Y))$. This ordinary category inherits monoidal, limit and colimit structure from $\CCat$. 
\begin{definition}[e.g.~\cite{markov-entropy}]\label{def:enriched-markov}
  A \emph{$\VCat$-enriched Markov category} is
  a $\VCat$-enriched symmetric monoidal category such that the monoidal unit is terminal: $\CCat(X, I)\cong 1$, and such that the underlying symmetric monoidal category is equipped with the structure of a Markov category (i.e.~a comonoid structure in the underlying ordinary category).

A $\VCat$-enriched Markov category is moreover \emph{distributive}
if it has $\VCat$-coproducts that distribute over the monoidal structure,
and such that the coproduct injections are deterministic, in the sense of the underlying ordinary category. 
\end{definition}
For any semicartesian category $\GCat$, recall the category of
functors $[\GCat\op,\Set]$. This extends $\GCat$ to a good `cosmos' for enrichment since   \begin{itemize}
\item $[\GCat\op,\Set]$ embeds $\GCat$ fully and faithfully (i.e. essentially as a full subcategory), via the Yoneda embedding
  $y(\gamma)=\GCat(-,\gamma)$.
\item $[\GCat\op,\Set]$ has all limits and colimits, computed pointwise.
\item $[\GCat\op,\Set]$ has a semicartesian structure such that the Yoneda embedding is a symmetric monoidal functor. This is given by Day convolution~\cite{day-convolution}, and has the following universal property: for $F, G, H\in[\GCat\op,\Set]$, to give a natural transformation $F\otimes G\to H$ is to give a natural family of functions ${F(\gamma)\times G(\varepsilon)\to H(\gamma\otimes \varepsilon)}$. 
\item $[\GCat\op,\Set]$ is moreover monoidal closed.
\end{itemize}
\begin{proposition}
  To give an $[\GCat\op,\Set]$-enriched distributive Markov category is to give a $\GCat$-graded distributive Markov category.
\end{proposition}
\begin{proof}[Notes]
  This follows from the characterization of locally $\GCat$-graded categories as $[\GCat\op,\Set]$-enriched categories (e.g.~\cite{wood-graded,pbl-graded,gavranovic-thesis}), and then translating Definition~\ref{def:enriched-markov} across this correspondence to arrive at Definition~\ref{def:graded-markov}. 
\end{proof}
The correspondence between graded monads and enriched monads is also well understood (e.g.~\cite{mu-flex-grade}).

\paragraph{Aside.}
Recent work by Perrone~\cite{markov-entropy} has considered enriched Markov categories to obtain an abstract view of the distance between probabilities, which allows for an abstract development of entropy.
Their enriching category~$\VCat=\mathsf{Div}$ in~\cite{markov-entropy} is indeed semicartesian.
The full theory of enriched Markov categories perhaps deserves a more thorough analysis.

 \section{A graded Markov category for imprecise probability}
 \label{sec:our-category}
 \newcommand{\RR}{\mathbb{R}}
We recall ordinary Markov categories for finite probability (\S\ref{sec:finstoch}). We then consider a generic construction for graded Markov categories and instantiate it in our setting, obtaining the graded Markov category $\MainCategory$ (for `Imprecise Probability', \S\ref{sec:maincategory}). We conclude this section with a worked example (\S\ref{sec:maincategory-example}). In the subsequent sections (\S\ref{sec:oplax-functor}--\ref{sec:max}) we relate this graded Markov category with convex sets of distributions. 
 
\subsection{Ordinary Markov categories for probability}
We recall the Markov category $\FinStoch$ of finite sets and stochastic maps between them.

In what follows, for notational simplicity,
we regard every finite set~$n$ with a given enumeration $n=\{a_1,\dots, a_{|n|}\}$.
  Also, we choose a particular singleton set $1$ (empty product), and for every pair of finite sets, we choose a product set $(m\times n)$
  and projections $m\leftarrow (m\times n)\rightarrow n$. (This will have cardinality $\#(m\times n)=\#m \times \# n$, but there are many isomorphic choices for the set.)
\label{sec:finstoch}
\begin{definition}\label{def:prob-vect}
  Let $n$ be a finite set.  A \emph{probability vector} $p\in\RR^n$ is an $n$-indexed sequence of non-negative numbers that sum to $1$. We write $D(n)$ for the set of probability vectors of length~$n$.
\end{definition}
The set $D(n)$ is always a convex set: for any $r\in [0,1]$ and $p,q\in D(n)$, the convex combination
$r\cdot p + (1-r)\cdot q$ is again a probability vector in $D(n)$.
We write $p +_r q$ as shorthand for such a convex combination.
Every probability vector in $D(n)$ arises via convex combinations of the Dirac vectors
$\delta_i$, for $i\in n$, where $\delta_1=(1,0,0,0\dots)$, $\delta_2=(0,1,0,0\dots)$ and so on.

A matrix of real numbers $f\in \RR^{n\times m}$ is called \emph{stochastic}
if each column is a probability vector.
This is equivalent to requiring that as a linear map, it preserves the property of being a probability vector,
i.e.~if $p\in D(n)$ then $(f\, p)\in D(m)$.
In fact, every function $D(n)\to D(m)$ that preserves convex structure arises from a stochastic matrix in this way.
We call such a function a \emph{convex map}.
\begin{definition}[e.g.~\cite{fritz}, Ex.~2.5]\label{def:finstoch}
  The category $\FinStoch$ of finite sets and stochastic maps has as objects finite sets, and morphisms $m\to n$ 
  stochastic matrices in $\RR^{n\times m}$. Composition is matrix multiplication, and the identity morphism is the unit diagonal matrix.

  This can be made into a symmetric monoidal category, with monoidal structure on objects given by products of finite sets. On morphisms, using the enumeration, we take the Kronecker product of
  matrices. It is semicartesian where the terminal object is~$1$ because there is a unique stochastic matrix with one row. This is moreover a Markov category, with $\copymap_n:n\to n\otimes n$ given by the three-dimensional diagonal (in $\RR^{(n\times n)\times n}$).

  The Markov category $\FinStoch$ moreover has a distributive structure. The coproduct of objects is given by disjoint union, i.e.~by adding cardinalities, and with this view the coproduct of morphisms forms block matrices (concatenating the columns).
\end{definition}
The monad view on $\FinStoch$ is as follows. First, we consider the embedding ${J:\FinSet\to \Set}$.
We then regard $D$ (Def.~\ref{def:prob-vect}) as a $J$-relative monad $D\colon \FinSet\to \Set$, 
which is affine and commutative. In fact, there is an ordinary monad $D'$ on $\Set$, comprising
finitely supported probability distributions (e.g.~\cite[Ch.~2]{bart-prob-book}), and $D=D'J$.
The distributive Markov category $\FinStoch$ can then be regarded as the Kleisli category for this relative monad.

\subsection{Choices of products} \label{sec:choices-products}
As noted above, the monoidal structure and copy maps of $\FinStoch$ depend on the choice
of products of sets. Categorically this does not matter, because it is all the same up to canonical isomorphism. But since the grades enter the syntax of the language (\S\ref{sec:graded-markov}), the choice of product can affect the syntax because identity maps can be omitted while isomorphisms cannot.
(Subtleties about the choice of structure are not unusual in categorical logic, e.g.~\cite{cgh-revisit}.)
We consider two choices:
\begin{description}
\item[Numerals:] Say that a \emph{numeral} is a finite set of the form $\{1,\dots,n\}\subseteq \NN$, for some natural number $n\in \NN$.
  \begin{itemize}
  \item We choose the empty product $1=\{1\}$ to be the unit numeral.
  \item For numerals $M$ and $N$ we choose the product to be the numeral for the integer multiple $|M|\times |N|$, with the projections coming from quotient and remainder. For all other finite sets, it does not matter which product structure we choose. 
  \end{itemize}
  Since numerals are closed under this choice of product structure, we could cut down to the skeletal full subcategory of $\FinStoch$ whose objects are numerals.
  This has the property that the unitors and associator isomorphisms are identity maps, but the symmetry isomorphisms are not.

  This choice of product is reminiscent of de Bruijn notation, where the indexing is derived from syntactic positions and reindexing is explicit, but weakening can be clumsy.
\item[Name-respecting choices:]
  Let $\Atoms$ be an infinite set of names.
  If $m\subseteq \Atoms$ is a finite set of names then the set of functions $2^m$ describe true/false assignments
  to those names, as in Knightian choices. 
  \begin{itemize}\item   We let the empty product be the set of functions $2^\emptyset$.
  \item   If $m,n\subseteq \Atoms$ are disjoint finite sets of names then we choose the product $2^m\times 2^n$ to be
    the set $2^{m\cup n}$. The projections $2^m\leftarrow 2^{m\cup n}\to 2^n$ are given by restricting the functions to subsets. For all other finite sets, it does not matter which product structure we choose.
  \end{itemize}
  This kind of choice of product structure has the property that for sets of the form $2^n$, the unitor isomorphisms are identities; and the symmetry and associativity isomorphisms
  \[2^{m}\times 2^{n}\cong 2^{n}\times 2^{m}
    \qquad\qquad
 2^{m}\times (2^{n}\times 2^p)\cong    (2^{m}\times 2^{n})\times 2^p
  \] are identity maps
  for disjoint $m$, $n$ and $p$.

  This choice of product is reminiscent of nominal techniques (e.g.~\cite{pitts-nom-book}). It is convenient for the language of \S\ref{sec:graded-pl} because under the disjointness assumption, all the coherence regradings in Proposition~\ref{eqn:graded-pl-sound} are identity functions, and so can be omitted. Moreover, the regrading in the affine law can often be elided because it can be inferred from the grades. (Again, formal grade inference algorithms are beyond the scope of this paper.)
\end{description}

\subsection{The graded Markov category $\MainCategory$}
\label{sec:maincategory}
We now introduce our graded Markov category for imprecise probability, $\MainCategory$ (Def.~\ref{def:imp}). We first introduce a general construction for graded Markov categories (Prop.~\ref{prop:paragraded}). This is a variation on the `para' construction~\cite{backprop}, also called monoidal indeterminates~\cite{hermida-tennent}. Via the connections between Markov categories and commutative affine relative monads (\S\ref{sec:monad-def}), it is equivalently a graded version of the reader monad transformer~\cite{monad-trans} of the finite distributions monad.
\begin{equation} \label{eq:graded-monad}
  T_{\gamma}(X) = [\gamma\Rightarrow D(X)]
\end{equation}
\begin{proposition}
  \label{prop:paragraded}
  Let $\GCat$ be a semicartesian subcategory of a distributive Markov category $\CCat$.
  There is a graded distributive Markov category with the same objects as $\CCat$ and with the hom-sets
  \[\CCat_\gamma(X,Y)=\CCat(\gamma\otimes X,Y)\qquad(\gamma\in\GCat)\text.\]
  The reindexing is given by composition: if $f\in\GCat(\varepsilon,\gamma)$ and $g\in \CCat_\gamma(X,Y)$
  \[f^*(g)=g\circ (f\otimes X) \in \CCat_\varepsilon(X,Y)\text.\]
For the composition of $f\in \CCat_\gamma(X,Y)$ and $g\in\CCat_\varepsilon(Y,Z)$,
\[
  (g\circ f)=\Big(\gamma\otimes \varepsilon\otimes X\cong \varepsilon\otimes \gamma \otimes X\xrightarrow{\varepsilon\otimes f}\varepsilon\otimes Y\xrightarrow g Z\Big)\in\CCat_{\gamma\otimes \varepsilon}(X,Z)\text.
\]
The monoidal product of $f\in \CCat_\gamma(X,X')$ and $g\in\CCat_\varepsilon(Y,Y')$ is given by
\[
  (f\otimes g)=\Big(\gamma\otimes \varepsilon\otimes X\otimes Y\cong \gamma \otimes X\otimes \varepsilon\otimes Y\xrightarrow{f\otimes g}X'\otimes Y'\Big)\in\CCat_{\gamma\otimes \varepsilon}(X\otimes Y,X'\otimes Y')\text.
\]
\end{proposition}
\begin{definition}
  A stochastic map $f\in\FinStoch(m,n)$ is \emph{surjective} if for every $j\in\{1\dots n\}$ there exists $i\in\{1\dots m\}$ such that
  $f_{i}$ is the Dirac distribution at $j$. In other words, the induced convex map $D(m)\to D(n)$ is surjective. 
  Let $\FinStochSurj$ be the category of finite sets and surjective stochastic maps.
  This is a semicartesian monoidal subcategory of $\FinStoch$.
\end{definition}
\begin{definition}\label{def:imp}
  The graded distributive Markov category $\MainCategory$ is the $\FinStochSurj$-graded version of $\FinStoch$, according to Proposition~\ref{prop:paragraded}.
\end{definition}
$\MainCategory$ is equivalently the graded Kleisli category of the monad \eqref{eq:graded-monad} over $\FinSet$ by the following isomorphisms, where regrading is done via the Kleisli extension
of the finite distributions monad.
\[
  \FinSet(X, [\gamma\Rightarrow D(Y)]) \cong \FinSet(\gamma\otimes X, D(Y)) \cong \FinStoch(\gamma\otimes X, Y) = \MainCategory_\gamma(X,Y).
\]
This graded distributive Markov category supports both finite probability and Knightian non-determinism.
\begin{itemize}
\item For binary probabilistic choice with bias $0.5$, we consider the morphism $\mlstinline{bernoulli}\in \MainCategory_1(1,2)$ given
  by the column vector $\begin{pmatrix} 0.5 \\ 0.5\end{pmatrix}$.
\item For a Knightian choice, we consider the morphism in $\mlstinline{knight}\in\MainCategory_2(1,2)$ given by
  the unit diagonal matrix. Formally this depends on the choice of product $\otimes$ (\S\ref{sec:choices-products}).

  For the name-respecting choice of product, for any $a\in\Atoms$ we let $\widehat a=2^{\{a\}}(\cong 2)$ and we have $\mlstinline{knight(a)}\in \MainCategory_{\widehat a}(1,2)$.
\end{itemize}
Thus $\MainCategory$ is a model of the language in \S\ref{sec:graded-pl}.

We can extend the above notions of probabilistic and non-deterministic choice between
elements of a finite set~$n$ by considering probability vectors (in $\MainCategory_1(1,n)$) and unit diagonal matrices (in $\MainCategory_n(1,n)$) respectively.
\paragraph{Remark:}
We could have considered a subcategory of $\FinStochSurj$ as the grading. One example is finite sets and (deterministic) surjective functions. Another example is the subcategory where the objects are of the form $2^A$ and where we only consider the surjections $2^B\to 2^A$  induced by injections $A\to B$ (connecting even closer with nominal sets~\cite{pitts-nom-book}; in this case, the semicartesian monoidal structure amounts to the disjoint union, $A\uplus B$.).
We leave for future work the question of to what extent the following results depend on this particular choice of grading. 
\subsection{Example calculation with $\MainCategory$}
\label{sec:maincategory-example}
In this example, to keep calculations brief, we use the numeral choice of product (\S\ref{sec:choices-products}). 
\begin{example} \label{exCalcs}
  Consider the scenarios from Listings~\ref{lst:prog1} and~\ref{lst:prog2} where we draw boolean values with Knightian uncertainty and from fair Bernoulli trials and combine them using different program logic.
  We denote outcomes as probability vectors of length three, representing the chance of $\rx$, $\gx$, and $\bx$, respectively. The program from Listing~\ref{lst:prog1} is the morphism
  \[
    \setlength\arraycolsep{5pt}
    (g + h) \circ f = \begin{pmatrix}
      1 & 0 \\ 0 & 0.5 \\ 0 & 0.5
    \end{pmatrix} \in \MainCategory_2(1,3),
  \]
  where $f$ denotes the conditional on the fair Bernoulli trial, $g$ and $h$ are the conditionals on the Knightian choices in each branch.
  \[
    \setlength\arraycolsep{5pt}
    f = \begin{pmatrix}
      0.5 \\ 0.5
    \end{pmatrix} \in \MainCategory_1(1,2) \qquad
    g = \begin{pmatrix}
      1 & 0 \\ 0 & 1 \\ 0 & 0
    \end{pmatrix} \in \MainCategory_2(1,3) \qquad
    h = \begin{pmatrix}
      1 & 0 \\ 0 & 0 \\ 0 & 1
    \end{pmatrix} \in \MainCategory_2(1,3)
  \]
  On the other hand, the program from Listing~\ref{lst:prog2} is the morphism
  \[
    \setlength\arraycolsep{5pt}
    (\widehat{\pi}_1^*(g) + \widehat{\pi}_2^*(h)) \circ f = \begin{pmatrix}
      1 & 0.5 & 0.5 & 0  \\ 0 & 0 & 0.5 & 0.5 \\ 0 & 0.5 & 0 & 0.5
    \end{pmatrix} \in \MainCategory_4(1,3),
  \]
  where $f$, $g$, and $h$ denote the same conditional statements, but now we lift the grading of $g$ and $h$ to $4$ via the projections to the Dirac distributions
  $\widehat{\pi}_1,\widehat{\pi}_2 \in \FinStochSurj(2\otimes 2, 2)$, to account for the decoupling of their Knightian uncertainties.
\end{example}

\section{First theorem: Relationship with the monad of convex sets of distributions}
\label{sec:oplax-functor}
In this section, we recall the properties of convex powersets of distributions. These form a monad~$\CP$ that has been proposed as a model for imprecise probability  (see also~\cite{mio-sarkis-vignudelli,mio-vignudelli,keimel-plotkin,bonchi-sokolova-vignudelli,jacobs-2008} and elsewhere). In contrast to our model $\MainCategory$ (\S\ref{sec:our-category}), it is not graded but also not commutative.
We connect our category $\MainCategory$ with convex powersets via the Kan extension method of~\cite{fritz-perrone-kan} (\S\ref{sec:kan-extensions}) and show that this yields an op-lax functor (\S\ref{sec:lax-functor}). The interpretation is that composition in our category gives tighter uncertainty bounds (Theorem~\ref{thm:unc-bounds}), which we illustrate in \S\ref{sec:oplax-example}.

We begin by recalling some basic properties of convex sets of distributions.

\newcommand{\image}{\mathsf{image}}
\begin{definition}
  A subset $S$ of $D(n)$ is \emph{convex} if
  it is closed under convex combinations: if $p,q\in S$ then for any $r\in[0,1]$ we have
  $p +_r q\in S$. 

A convex subset $S$ of $D(n)$ is \emph{finitely generated} if there is a finite sequence
$p_1\dots p_m\in S$ such that every element of $S$ can be achieved by convex combinations of the $p_i$'s.
In other words, $S=\{q\cdot (p_1\dots p_m)~|~q\in D(m)\}$, with the $p_i$'s regarded as column vectors and $q$ regarded as a row vector. 
\end{definition}
\begin{lemma}
  For any convex map $f:D(m)\to D(n)$ between the sets of probability vectors, 
  the image of $f$ is a convex subset of $D(n)$.

  Moreover, such convex subsets of $D(n)$ are finitely generated, and every finitely generated convex set arises in this way.
  \label{lemma:conveximage}
\end{lemma}
\begin{proof}
  Suppose $q,q'\in\image(f)$, and let $r\in[0,1]$.
  So we must have $p,p'\in D(m)$ such that $f(p)=q$ and $f(p')=q'$. 
  Then
  \[q +_r q' =
    f(p) +_r f(p') =
    f(p +_r p')\text,
  \]
  the last step because $f$ is a convex map, and so we see that $q +_r q'\in \image(f)$.

  The set $\image(f)$ is generated by $f(\delta_i)$ for $i=1\dots m$.
  Conversely if a set $S$ is generated by $p_1\dots p_m$, regarded as column vectors, then
  the matrix $(p_1\dots p_m)\in \FinStoch(m,n)$ determines a map $f:D(m)\to D(n)$ such that
  $\image(f)=S$.
\end{proof}
\subsection{Convex powersets of distributions}
\label{sec:convpowerset}
We write $\CPf(n)$ for the finitely generated convex subsets of $D(n)$. 
It supports convex combinations: if $r\in[0,1]$ and $S,T\in \CP(n)$ then
\[
  S +_r T\defeq \{p +_r q~|~p\in S,q\in T\}\in \CP(n)\text.
\]
There is moreover an ordering given by subset, and the join is
a convex closure of the union:
\[
  S\vee T \defeq \{p +_r q~|~r\in[0,1],p\in S,q\in T\}\text.
\]

\begin{proposition} \label{prop:rs2cp}
  There is a family of functions $\phi_{m,n}:\MainCategory_m(1,n)\to \CPf(n)$, that takes
  $f\in\MainCategory_m(1,n)$ to its image $\image(f)\in\CPf(n)$, and the family is natural in $m\in\FinStochSurj$ and
  $n\in \Set$.
\end{proposition}
\begin{proof}
  First, the fact that the image of $f$ is convex is Lemma~\ref{lemma:conveximage}.
  For naturality in $m$, suppose $g\in \FinStochSurj(m',m)$.
  Then naturality in $m$ amounts to the fact that
  \[\image(f\circ g) = \image(f)
  \]
  which is true since $g$ is surjective.
  For naturality in $n$, suppose $h \in \Set(n,n')$. Then naturality amounts to the fact that
  \[\image(D(h)\circ f) = \CPf(h)(\image(f))
  \]
  which is true because taking an image of $f$ after postcomposition with $D(h)$ is the same as a pointwise application of $D(h)$ to the image of $f$.
\end{proof}
\subsection{Connection to Kan extensions} \label{sec:kan-extensions}
Fritz and Perrone~\cite{fritz-perrone-kan,fp-metric-monad} propose a method to extract a canonical monad from a graded monad, by taking the left Kan extension.
They provide criteria for when this process works and induces a monad morphism.
This process \emph{cannot work} entirely for our situation, for the following reason.
First we note that we can interpret
\lstinline|bernoulli| and \lstinline|knight| as the following elements of $\CP(2)$ (considering sets up to their convex closure):
\begin{itemize}
\item \lstinline|bernoulli| is $\{(\begin{smallmatrix}1\\0\end{smallmatrix})\}+_{0.5}\{(\begin{smallmatrix}0\\1\end{smallmatrix})\} \ =\  \{(\begin{smallmatrix}0.5\\0.5\end{smallmatrix})\}$;
\item \lstinline|knight| is $\{(\begin{smallmatrix}1\\0\end{smallmatrix})\}\vee\{(\begin{smallmatrix}0\\1\end{smallmatrix})\}\ =\  \{(\begin{smallmatrix}1\\0\end{smallmatrix}),(\begin{smallmatrix}0\\1\end{smallmatrix})\}$.
\end{itemize}
This construction $\CP$ extends to a monad on $\Set$~\cite{jacobs-2008}. Therefore, we can follow through the derivation of Figure~\ref{fig:derivation} to see that the $\CP$ monad \emph{cannot} be commutative (apparently contradicting \cite[Lemma~5.2]{jacobs-2008}) since the convex sets in Figure~\ref{fig:tris}b--\ref{fig:tris}c are different.
(For another argument, note that $\CP$ contains two binary idempotent symmetric operations, $\vee$ and $+_{0.5}$, and see our discussion on the Eckmann-Hilton-like obstacle~\eqref{eqn:eckmann-hilton}
in \S\ref{sec:alg-perspective}.)

By contrast, $\MainCategory$ (\S\ref{sec:maincategory}) \emph{does} satisfy our desiderata (\S\ref{sec:intro:desiderata}). So there cannot be a monad morphism between $\MainCategory$ and $\CP$. Nonetheless, the Kan extension of our graded Markov category $\MainCategory$, regarded as a graded monad via Proposition~\ref{prop:markov-monad}, does give the finitely-generated convex powerset monad~$\CPf$ as a functor, just not as a monad.

\begin{proposition} \label{prop:kan-ext}
  The family $\phi_{m,n}:\MainCategory_m(1,n)\to \CPf(n)$ exhibits
  $\CPf:\FinSet\to \Set$ as the Kan extension of \[\MainCategory_{(-)}(1,=):\FinStochSurj\op \to[\FinSet,\Set] \] along the unique functor $\FinStochSurj\op\to 1$.
  \[\begin{tikzcd}
    \FinStochSurj\op & 1 \\
    & {[\FinSet,\Set]}
    \arrow["{\CPf}", from=1-2, to=2-2]
    \arrow["{!}", from=1-1, to=1-2]
    \arrow["{\MainCategory_{(-)}(1,=)}"', from=1-1, to=2-2]
  \end{tikzcd}\]
\end{proposition}
\begin{proof}
  Kan extensions in $[\FinSet,\Set]$ can be computed pointwise,
  and for any $n\in\FinSet$ the Kan extension of $\MainCategory_{(-)}(1,n):\FinStochSurj\op\to\Set$
along $\FinStochSurj\op\to 1$ is simply the colimit of the functor. Thus it suffices to show that the canonical function
\[\Phi:\colim_{m\in\FinStochSurj\op}\MainCategory_m(1,n)\to \CPf(n)\]
(induced by $\phi$) is a bijection.
This function $\Phi$ is given by $\Phi[m,f\in \FinStoch(m,n)] = \image(f)$.
To see that it is surjective we recall that every finitely generated convex set is the image of some convex function $D(m)\to D(n)$
(Lemma~\ref{lemma:conveximage}).
To see that it is injective we suppose that $\image(f)=\image (f')$, for $f\in\FinStoch(m,n)$ and $f'\in \FinStoch(m',n)$.
We must show that $[m,f]=[m',f']$ in the colimit. It suffices to find $m''$ with $h\in\FinStoch(m,n)$
and surjections $g\in\FinStochSurj(m,m'')$ and $g'\in\FinStochSurj(m',m'')$ such that the following diagram commutes:
\[\begin{tikzcd}[column sep=scriptsize]
	m && {m'} \\
	& {m''} \\
	& n
	\arrow["g", from=1-1, to=2-2]
	\arrow["{g'}"', from=1-3, to=2-2]
	\arrow["f"', from=1-1, to=3-2, bend right=20]
	\arrow["{f'}", from=1-3, to=3-2, bend right=-20]
	\arrow["h"', from=2-2, to=3-2]
\end{tikzcd}\]
The finitely generated convex set $\image(f)=\image(f')$ must have a unique convex hull, and we let $m''$ be the number of extremal points of the convex hull, which are uniquely determined. We construct $g$ by noting that $f(i)$ must be a convex combination from the $m''$ extremal points, so we let $g(i)$ be the probability vector corresponding to that combination. We construct $g'$ from $f'$ similarly. To see that~$g$ is surjective we note that since $f$ is surjective onto its image we must have points in $m$ that map onto the extremal points, and hence onto all the points of $m''$ via $g$. Similarly, $g'$ is surjective. 
\end{proof}
\newcommand{\ext}{\mathsf{ext}}
\subsection{An op-lax functor and tighter uncertainty bounds} \label{sec:lax-functor}
\begin{definition}
  The construction $\CPf$ extends to a relative monad. 
  The unit morphism $\eta_n:n\to \CPf(n)$ picks out the singleton set containing the Dirac vector, $\eta_n(i)=\{\delta_i\}$.
  The Kleisli extension takes a function $f:m\to \CPf(n)$ to a function $f^*:\CPf(m)\to \CPf(n)$ given by
  \[
    f^*(X) = \bigvee_{x\in \ext(X)} \sum_{i\in m} x_{i} \cdot f(i);
  \]
  where $\ext$ takes the extreme points of the finitely generated convex subset.
\end{definition}
From this structure, we build a Kleisli category as usual.
\begin{itemize}
\item The objects of $\Kl(\CPf)$ are natural numbers.
\item The morphisms $m\to n$ are functions $m\to \CPf(n)$.
\item The identity morphism is the unit $\eta$. Composition of $g$ and $f$ is given by $g^* \circ f$. 
\end{itemize}
In fact, this category is order-enriched. That is to say, the hom-sets $\Kl(\CPf)$ have a natural partial order structure given by $f\leq g$ if for all $i$, $f(i)\subseteq g(i)$. Composition is thus monotone. 

We now extend the quotient of Proposition~\ref{prop:rs2cp} to an identity-on-objects op-lax functor
$\MainCategory\to \Kl(\CPf)$. 
\begin{theorem} \label{thm:unc-bounds}
  Consider the assignment of a morphism $f\in \MainCategory_\gamma(m,n)$ to
  $R(f) : m\to \CPf(n)$ given by $R(f)(i)=\image(f(-,i))$. This defines an op-lax functor
  \[\MainCategory\to \Kl(\CPf)
  \]
\end{theorem}
\begin{proof}[Proof notes.]
  It is straightforward that $R(\id)=\id$.
  It remains to show that $R(g\circ f)\subseteq R(g)\circ R(f)$.
  Since we will show that $R$ preserves finite coproducts, it is sufficient to suppose that the domain of $f$ is $1$. 
  So consider $f\in \MainCategory_\gamma(1,m)$ and $g\in\MainCategory_\varepsilon(m,n)$.
  So $(g\circ f)\in\FinStoch(\gamma\times \varepsilon,n)$. 
  We must show that for all $(i,j)\in (\gamma\times \varepsilon)$, the probability vector $(g\circ f)(i,j)$
  is in
  \[(R(g)\circ R(f))() = R(g)^*(\image(f))\in\CPf(n)\text.\]
  To show this, we note that the grade of $(g\circ f)$ is $(\gamma\times \varepsilon)$, but we can also consider an alternative kind of composite
  $(g\ast f)$ with a bigger grade $(\gamma\times \varepsilon^m)$. This is given by
  \[(g\ast f)=\Big(\gamma\times \varepsilon^m\xrightarrow f m\times \varepsilon^m \xrightarrow {(\mathit{eval}, \pi_1)} \varepsilon\times m
    \xrightarrow g n\Big)\text;\]
  where the middle arrow is the evident function between sets regarded as a stochastic matrix.
  Contrast with
  \[(g\circ f)=\Big(\gamma\times \varepsilon\xrightarrow f m\times \varepsilon \xrightarrow {\mathit{swp}} \varepsilon\times m
    \xrightarrow g n\Big)\text.\]
  The function $(\gamma\times \varepsilon)\to (\gamma\times \varepsilon^m)$ that copies $\varepsilon$ is an injection and exhibits
  \[\image(g\circ f)\subseteq \image(g\ast f)
  \]
  Moreover, we have that \[\image(g\ast f)=R(g)\circ R(f)\text.\] The intuitive point is that in $(g\ast f)$, for each possible intermediate~$m$ we are allowed to use different choices of~$\varepsilon$, but in $(g\circ f)$, each possible intermediate~$m$ will use the same choices of~$\varepsilon$.

  To see that $R$ preserves coproducts we note that on objects it is immediate, and expanding the definitions shows that the coproduct injections and copairings are exactly preserved by $R$. 
\end{proof}
(Here, we are regarding $\MainCategory$ with discrete order enrichment but non-trivial local grading, and
$\Kl(\CPf)$ with non-trivial order enrichment but trivial local grading.
There may be interesting ways to unify the two different enrichments.)
\subsection{Discussion and example of tighter uncertainty bounds}\label{sec:oplax-example}

\begin{example}
  We again revisit the scenarios from Examples~\ref{exA} and~\ref{exB} where boolean values are drawn with Knightian uncertainty or from fair Bernoulli trials and combined using different program logic.
  Consider the morphism denoting a fair Bernoulli trial ($f$ from Example~\ref{exCalcs}),
  \[
    \setlength\arraycolsep{5pt}
    f = \begin{pmatrix}
      0.5 \\ 0.5
    \end{pmatrix} \in \MainCategory_1(1,2),
  \]
  and a morphism that employs Knightian uncertainty on each of its inputs ($g+h$ from Example~\ref{exCalcs}),
  \[
    g = \left(\begin{matrix}
      1 & 0 \\ 0 & 1 \\ 0 & 0
      \end{matrix} \ \middle| \ \begin{matrix}
      1 & 0 \\ 0 & 0 \\ 0 & 1
    \end{matrix}\right) \in \MainCategory_2(2,3).
  \]
  Then, considering sets up to their convex closure, $R(f) : 1 \to \CPf(2)$ maps the singleton set to
  \[
    \left\{\begin{pmatrix} 0.5 \\ 0.5\end{pmatrix}\right\},
  \]
  $R(g) : 2 \to \CPf(3)$ maps the two-element set to
  \[
    \left\{\begin{pmatrix} 1 \\ 0 \\ 0 \end{pmatrix}, \begin{pmatrix} 0 \\ 1 \\ 0 \end{pmatrix}\right\} \text{ and }
    \left\{\begin{pmatrix} 1 \\ 0 \\ 0 \end{pmatrix}, \begin{pmatrix} 0 \\ 0 \\ 1 \end{pmatrix}\right\},
  \]
  and (following Example~\ref{exCalcs}) $R(g\circ f) : 1 \to \CPf(3)$ maps the singleton set to
  \[
    \left\{\begin{pmatrix} 1 \\ 0 \\ 0 \end{pmatrix}, \begin{pmatrix} 0 \\ 0.5 \\ 0.5 \end{pmatrix}\right\}.
  \]
  This is the convex subset in Figure~\ref{fig:tris}b if we consider the probability vectors as giving the corresponding chances of outcomes $\rx$, $\gx$, and $\bx$.

  On the other hand, by composing $g$ with $f$ after mapping them into $\mathrm{Kl}(\CPf)$, we lose the ability to distinguish which outcomes were related via the
  same Knightian choices. So the morphism $R(g)\circ R(f) : 1 \to \CPf(3)$ maps the singleton element to
  \[
    \left\{\begin{pmatrix} 1 \\ 0 \\ 0 \end{pmatrix}, \begin{pmatrix} 0.5 \\ 0 \\ 0.5 \end{pmatrix}, \begin{pmatrix} 0.5 \\ 0.5 \\ 0 \end{pmatrix}, \begin{pmatrix} 0 \\ 0.5 \\ 0.5 \end{pmatrix}\right\},
  \]
  which is the convex subset given in Figure~\ref{fig:tris}c. Thus, $R(g\circ f) \subsetneq R(g)\circ R(f)$.
\end{example}

Therefore, by accounting for corresponding choices of Knightian uncertainty within morphism compositions, our category $\MainCategory$ obtains tighter bounds on the imprecise probabilities.

\section{Second theorem: Maximality as a compositional theory}
\label{sec:max}
In Proposition~\ref{prop:rs2cp} we gave a family of maps $\phi$ that convert our compositional imprecise probability into convex sets of probability distributions.
These maps are not injective, and in this sense the model of $\MainCategory$ is intensional. This raises a question of whether we could have made a less intensional model than $\MainCategory$ while still maintaining Desiderata~1 and~2 and the connection to convex sets of distributions. In Theorem~\ref{thm:maximality} we answer this question negatively in the following sense: we cannot quotient
the hom-sets of $\MainCategory$ without either losing the connection with convex sets (and hence statistics) or losing the monoidal or distributive structure (and hence the compositionality desiderata of \S\ref{sec:intro:desiderata}). In this way, $\MainCategory$ is maximal.
\begin{definition}
  Let $\GCat$ be a semicartesian category. Let $\CCat$ and $\DCat$
  be $\GCat$-graded distributive Markov categories.
  A \emph{graded distributive Markov functor} $F:\CCat\to \DCat$ is given by
  a mapping from the objects of $\CCat$ to the objects of $\DCat$ and
  a family of mappings from $\CCat_\gamma(X,Y)\to \DCat_\gamma(F(X),F(Y))$,
  strictly preserving the composition, monoidal and coproduct structure,
  and the copy maps. 
\end{definition}
\paragraph{Aside.} In view of \S\ref{sec:markov-def}, we note that a graded distributive Markov functor is the same thing as the existing notion of strict distributive monoidal functor between distributive monoidal enriched categories (e.g.~\cite{kelly}), together with the requirement that the copy maps are preserved, which is in common with the ordinary Markov category literature~\cite{fritz}. We could also formulate this in terms of monad morphisms, following \S\ref{sec:monad-def}.
\begin{theorem} \label{thm:maximality}
  Let $\CCat$ be $\FinStochSurj$-graded distributive Markov category
  with a graded distributive Markov functor
  $F:\MainCategory \to \CCat$ and a
  natural family of functions \[\psi_{m,n}:\CCat_m(1,n)\to \CPf(n)\] such that
  \[\phi_{m,n}:\MainCategory_m(1,n)\to \CPf(n)\]
  (Proposition~\ref{prop:rs2cp}) factors through $\psi$.
  Then $F$ is faithful: if $F(f)=F(g)$ in $\CCat$ then also $f=g$ in $\MainCategory$.
\end{theorem}
\begin{proof}
  Since $F$ preserves finite coproducts, it is sufficient to suppose the domain of $f$ and $g$ is $1$.
  That is, let $f,g \in \MainCategory_m(1,n)$ and suppose $\phi_{m,n}$ factors as
  \[
    \MainCategory_m(1,n) \xrightarrow{F_{m,n}} \CCat_m(1,n) \xrightarrow{\psi_{m,n}} \CPf(n).
  \]
  Let $d \in \MainCategory_m(1,m)$ be the evident tuple of Diracs.
  Define $\iota \in \MainCategory_1(m, n+m)$ and
  $\jmath \in \MainCategory_1(n, m+n)$ as the
  lifting of the injections $m \to m+n \leftarrow n$
  via postcomposition with the unit of $D$.
  Since $F$ is a graded distributive Markov functor
  and $F_{m,n}(f) = F_{m,n}(g)$,
  \[
    F_{m,m+n}(\jmath \circ f +_{0.5} \iota \circ d)
      = F_{m,m+n}(\jmath \circ g +_{0.5} \iota \circ d)
  \]
  where for $h, k: X \to D(Y)$, we define $(h+_{r} k)(x) = h(x) +_{r} k(x)$. Applying $\psi$ gives
  \begin{equation} \label{eq:convex_hulls}
    \phi_{m,m+n}(\jmath \circ f +_{0.5} \iota\circ d)
      = \phi_{m,m+n}(\jmath \circ g +_{0.5} \iota\circ d).
  \end{equation}
  Now, for all $i\in m$, $(\jmath \circ f +_{0.5} \iota \circ d)(i)$ are independent
  because they each use a different dimension. They are all extremal vertices on the convex
  hull $\phi_{n,m+n}(\jmath \circ f +_{0.5} \iota \circ d)$. Moreover, they must be the
  same vertices as $(\jmath \circ g +_{0.5} \iota \circ d)(i)$ for respective $i\in m$
  because the convex hulls are the same~\eqref{eq:convex_hulls}. Therefore,
  \[
    \jmath \circ f +_{0.5} \iota \circ d = \jmath \circ g +_{0.5} \iota \circ d.
  \]
  We can recover $f$ and $g$ as for any $i \in m$ and $j \in n$,
  \begin{align*}
    f(i)(j) = 2\times (\jmath \circ f +_{0.5} \iota \circ d)(i)(j), \\
    g(i)(j) = 2\times (\jmath \circ g +_{0.5} \iota \circ d)(i)(j).
  \end{align*}
  So $f = g$.
\end{proof}

\subsection{Remarks on different approaches to quotients} \label{sec:quotients}
We briefly remark on a different approach to quotients in a locally graded category.
This is a general method but connects to our language as follows. Notice that an open term in our language contains both names and variables: names for Knightian choices, and free variables standing for ordinary values that might be substituted later. There are two ways to quotient the names away, depending on how we order the quantifiers: 
\begin{description}
\item[($\forall\exists$) The approach leading to the $\CP$ monad (Theorems~\ref{thm:unc-bounds} and~\ref{thm:maximality}):] We could equate two open terms if, for every valuation of the free variables, there is a regrading that equates them.
  This violates Desideratum~\ref{des:1} but gives rise to an op-lax functor (Theorem~\ref{thm:unc-bounds}).
\item[($\exists\forall$) Alternative quotient:] We could equate two open terms if there is a regrading such that for every valuation of the free variables they are made equal. This violates Desideratum~\ref{des:2}, as we discuss below.
\end{description}

For closed terms with no free variables, the two approaches are the same and give rise to a convex set of distributions (Prop.~\ref{prop:rs2cp}).

The `alternative quotient' approach ($\exists\forall$) amounts to the `connectedness' quotient in~\cite[Def.~2.1]{hermida-tennent}. This does not satisfy Desideratum~\ref{des:2} (commuting if-then-else). Informally, it would allow us to work up to a different regrading on the `then' branch versus the `else' branch, which leads to the inconsistency. More formally, the construction of~\cite{hermida-tennent} does not yield a category with coproducts in general. For this reason, this alternative quotient approach is not a counterexample to Theorem~\ref{thm:maximality}. Nonetheless, it could be a useful approach in a metalanguage for combining models that do not need a general if-then-else construction. 

\section{Context and related work} \label{sec:notes-about-context}

Our focus here is on imprecise probability. This can be thought of as a kind of non-determinism. Specifically, we are combining the effects of probabilistic choices (\lstinline|bernoulli|) with non-deterministic choices arising from unknown probability distributions (\lstinline|knight|). 
There is a broader general interest in non-determinism and its combination with probability.
Non-determinism arises in many semantic situations beyond the motivation from unknown probabilities in imprecise probability.

\paragraph{Comparison with abstraction/refinement.} One arguably different motivation is in program abstraction and refinement: there, one describes a problem by writing a non-deterministic program that solves it in a non-determined way; the problem is then solved by refining that non-deterministic program. 
When the refined program is still probabilistic, the mathematical analysis is similar to imprecise probability, and for instance illustrations essentially the same as Figure~\ref{fig:tris} appear in work on the refinement of probabilistic programs~\cite[Fig.~6.4.2, Fig.~6.5.1]{mciver-morgan-book}. However, the motivation is arguably different and the desiderata (\S\ref{sec:intro:desiderata}) may be less relevant in program refinement. 

\paragraph{Contrast with random sets and random bags}
An arguably different kind of non-determinism appears where there are many appropriate results that we want to collect. In this sense, for instance, database queries are non-deterministic if they return multiple results, and Prolog is non-deterministic. When combined with probability, this leads more naturally to random sets or random bags, which contrast the sets of distributions shown in Figure~\ref{fig:tris}. Random bags do arise in probabilistic databases and point process theory. These applications have been considered from a monad perspective~\cite{dash-prob-db,dash-pp}, and the monads have long been discussed (e.g.~\cite{dsp-layer,keimel-plotkin,dash-thesis,kozen-silva,varacca-winskel,jacobs-bags}). 

Broadly, there is a contrast between \emph{sets of distributions}, which arise in imprecise probability and program refinement, and \emph{distributions over sets}, which appear in probabilistic databases and point processes.

\subsection{Relationship with work on distributive laws of monads}
There is a large literature on finding elegant explanations for combining existing monads for probability and non-determinism, exploring distributive laws of monads (e.g.~\cite{petrisan-sarkis,goy-petrisan,diaa,acg-diagrams-prob-nondet,varacca-winskel,bonchi-sokolova-vignudelli-B,dsp-layer,kozen-silva}).
In fact, the reader monad transformer that we use here is a distributive law of monads. However, our emphasis and motivation are from the commutativity desiderata (\S\ref{sec:intro:desiderata}) rather than distributivity issues.
We note that commutativity and affinity imply some distributive behaviour, allowing us to reorder computations into non-determinism followed by probability or vice versa. In~\eqref{distributivity} below, we derive a typical distributivity equation from just the basic desiderata. But that is not quite the same as a distributive law of monads, which is a very specific equational requirement~\cite[Def.~3, Thm.~5]{pirog-staton}; in some situations, it amounts to a weak distributive law of monads~\cite{goy-petrisan}.

Even when there is a distributive law between commutative monads,
the resulting composite monad need not be commutative. Indeed both the random bags monad~\cite{dash-thesis} and the powerdomain of indexed valuations monad~\cite{varacca-winskel} arise from distributive laws between commutative monads, but neither composite monads are commutative. This suggests that given our desiderata (\S\ref{sec:intro:desiderata}), distributive laws of monads are not necessarily the best starting point. 

\subsection{Algebraic perspective on probability and non-determinism} \label{sec:alg-perspective}
Some previous work on combining probability and non-determinism takes the perspective of algebraic theories. Our desiderata (\S\ref{sec:intro:desiderata}) can be viewed from the point of view of algebraic theories, via algebraic effects (e.g.~\cite{pp-notions}), which we now briefly explore.
We define two binary operations:
\[(t\probplus u) \defeq  \mlstinline{if bernoulli then t else u}
  \qquad
  (t\ndplus u) \defeq \mlstinline{if knight then t else u}
\]
Regarding Desideratum~\ref{des:1}, commutativity means that each operation is a homomorphism for the other:
  \begin{equation}
    \label{eqn:commutativity}
    \begin{array}{c@{\,}c@{\,}c@{\,}c@{\,}c@{\,}c@{\,}ccc@{\,}c@{\,}c@{\,}c@{\,}c@{\,}c@{\,}c}
(s &\ndplus& t)&\ndplus &(u&\ndplus &v)&=&
  (s &\ndplus &u)&\ndplus& (t&\ndplus &v)\\[6pt]
  (s &\probplus &t)&\probplus &(u&\probplus &v)&=&
  (s &\probplus &u)&\probplus &(t&\probplus &v)
\\[6pt]  (s &\ndplus &t)&\probplus &(u&\ndplus &v)&=&
  (s &\probplus& u)&\ndplus& (t&\probplus &v)
\end{array}\end{equation}
and affinity says that
\begin{equation}\label{eqn:affine}
  t\ndplus t=t\qquad\text{and}\qquad t\probplus t=t\text.
  \end{equation}
  In general, an algebraic structure satisfying commutativity and affinity is recognized as an abstract theory of
  probability, being the algebraic counterpart to Markov categories, and is sometimes called a `mode'~\cite{modes}. 

  Desideratum~\ref{des:2} is always assumed in algebraic effects. 

From these five axioms \eqref{eqn:commutativity}--\eqref{eqn:affine}, we can derive distributivity equations:
\begin{equation}\label{distributivity}\begin{array}{ccccc}
    t \ndplus (u \probplus v)&\stackrel{\text{affinity}}=&
  (t \probplus t) \ndplus (u \probplus v)&\stackrel{\text{commutativity}}=&
      (t \ndplus u) \probplus (t \ndplus v)\\
    t \probplus (u \ndplus v)&\stackrel{\text{affinity}}=&
  (t \ndplus t) \probplus (u \ndplus v)&\stackrel{\text{commutativity}}=&
  (t \probplus u) \ndplus (t \probplus v)
\end{array}\end{equation}
The first equation in~\eqref{distributivity} instantiates to 
\[\rx \ndplus (\gx \probplus \bx)=(\rx\probplus \rx) \ndplus (\gx \probplus \bx)
=      (\rx \ndplus \gx) \probplus (\rx \ndplus \bx)
\]
which is the problematic derivation of Figure~\ref{fig:derivation} written in algebraic form.
If we have an intuition of $\ndplus$ as a Minkowski sum (see \S\ref{sec:convpowerset}) then the left-hand side appears to be Figure~\ref{fig:tris}b and the right-hand side is Figure~\ref{fig:tris}c, especially since
the latter can be further rearranged using the commutativity and affinity laws to 
\[  (\rx\ndplus (\gx \probplus \rx) ) \ndplus ((\rx \probplus \bx) \ndplus (\gx \probplus\bx))
\]
which enumerates the four extreme points of Figure~\ref{fig:tris}c.

The symmetry laws
\begin{equation}t\probplus u=u\probplus t
  \label{eqn:symmetryprobplus}\end{equation}
\begin{equation} t\ndplus u =u\ndplus t 
  \label{eqn:symmetryvee}\end{equation}
are also desirable, but we note that it is already known that Desiderata~1  (\eqref{eqn:commutativity}--\eqref{eqn:affine}) are incompatible with them,
since together they imply
that the two binary operations ($\ndplus$, $\probplus$) are equal! (e.g.~\cite{eckmann-hilton-nlab}):
\begin{equation}\label{eqn:eckmann-hilton}
  t\ndplus u=(t\ndplus u)\probplus (t\ndplus u)=
  (t\ndplus u)\probplus (u\ndplus t)=
  (t\probplus u)\ndplus(u\probplus t)=
  (t\probplus u)\ndplus(t\probplus u)=
  t\probplus u
\end{equation}
We can regard the various previous algebraic works as different ways to avoid this Eckmann-Hilton-like obstacle~\eqref{eqn:eckmann-hilton}, by omitting various equations:
\begin{itemize}
\item The point process monad~\cite{keimel-plotkin,kozen-silva,jacobs-bags,dash-thesis} omits the last commutativity law~\eqref{eqn:commutativity} and the affine law for~$\ndplus$~\eqref{eqn:affine};
  it keeps the first distributivity law of~\eqref{distributivity}.
\item The powerdomain of indexed valuations~\cite{varacca-winskel} also omits the last commutativity law~\eqref{eqn:commutativity} but omits the affine law~\eqref{eqn:affine} for $\probplus$ instead of for $\ndplus$;
  it keeps the first distributivity law of~\eqref{distributivity}.
\item The convex powerset monad (e.g.~\cite{bonchi-sokolova-vignudelli-B}) also drops the last commutativity law but
  it keeps the \emph{second} distributivity law of~\eqref{distributivity}.
\end{itemize}
As we now explain, our approach can be viewed as a graded version of the algebraic perspective, in which we have \emph{all the equational laws}.

\subsection{Graded algebraic perspective} \label{sec:graded-alg}

We can look at our graded category $\MainCategory$ from this algebraic perspective,
using the graded algebraic theories of Kura~\cite[\S3.1]{kura-graded-algebraic}. We now summarize this angle. 
From the algebraic perspective, the proposal in this paper is to name the Knightian branching, which amounts to having named binary operators~$\ndplus_{a_1}$, $\ndplus_{a_2}$, \dots.
To be precise, in the formalism of \cite[\S3.1]{kura-graded-algebraic} we have a graded signature
\[\probplus \in \Sigma_{2,\emptyset}\qquad\ndplus_a\in \Sigma_{2,2^{\{a\}}}
\]
where $a\in \Atoms$. That is, we have a binary operation $\probplus$ with empty grade, and
binary operations $\ndplus_a$ with grade $2^{\{a\}}$. 
From this we define a grammar for terms $T_g(X)$ with variables in $X$ at grade $g$
\cite[Def.~10]{kura-graded-algebraic}:
\[
  \begin{prooftree}
    x\in X
    \justifies
    x\in T_{2^\emptyset}(X)\end{prooftree}
\qquad
\begin{prooftree}
  t,u\in T_g (X)\justifies
  t\probplus u\in T_g(X)
\end{prooftree}
\qquad
\begin{prooftree}
  t,u\in T_g (X)\justifies
  t\ndplus_a u\in T_{2^{\{a\}}\otimes g}(X)
\end{prooftree}
  \qquad
  \begin{prooftree}
    t\in T_g(X)
    \justifies
    c_f(t)\in T_h(X)
  \using \mbox{\small $f\colon h\to g$}\end{prooftree}
\]
Here `$c_f(t)$' indicates regrading along $f\colon h\to g$.
Henceforth we elide $c_f$ where $f$ is a `weakening' regrading, i.e.~a canonical projection
$2^{S}\to 2^{S'}$ for $S'\subseteq S\subset\Atoms$. 

We can consider analogues of all the laws~\eqref{eqn:commutativity}--\eqref{eqn:symmetryvee} above, in this graded setting:
\begin{equation}  \label{eqn:commgraded}  \begin{array}{c@{\,}c@{\,}c@{\,}c@{\,}c@{\,}c@{\,}ccc@{\,}c@{\,}c@{\,}c@{\,}c@{\,}c@{\,}c}
(u &\ndplus_b& v)&\ndplus_a &(x&\ndplus_b &y)&=&
  (u &\ndplus_a &x)&\ndplus_b& (v&\ndplus_a &y)\\[6pt]
  (u &\probplus &v)&\probplus &(x&\probplus &y)&=&
  (u &\probplus &x)&\probplus &(v&\probplus &y)
\\[6pt]  (u &\ndplus_a &v)&\probplus &(x&\ndplus_a &y)&=&
  (u &\probplus& x)&\ndplus_a& (v&\probplus &y)
                                          \end{array}\end{equation}
                                        \begin{align}
  x\probplus x&=x
                \\
  x\ndplus_a x&=x
  \\
  x\probplus y&=y\probplus x\hspace{-3cm}
                \\
x\ndplus_a y &= c_{\neg}(y\ndplus_a x)\hspace{-3cm}
         \label{eqn:assocgraded}
\end{align}
where $\neg:2^{\{a\}}\to 2^{\{a\}}$ is the non-trivial bijection.

Revisiting the example of Figure~\ref{fig:derivation}, we can use the graded deductive system of~\cite[Thm.~1]{kura-graded-algebraic} to derive:
\[
  \rx \ndplus_{a_1} (\gx \probplus \bx)=
  (\rx \probplus \rx) \ndplus_{a_1} (\gx \probplus \bx)=
  (\rx \ndplus_{a_1} \gx) \probplus (\rx \ndplus_{a_1} \bx)\text.
\]
yet it is consistent to assume
\[(\rx \ndplus_{a_1} \gx) \probplus (\rx \ndplus_{a_1} \bx)
\quad\neq\quad 
  (\rx \ndplus_{a_1} \gx) \probplus (\rx \ndplus_{a_2} \bx)\text.
\]
with different names on the right-hand side. These two terms correspond to the difference between Figures~\ref{fig:tris}b and~\ref{fig:tris}c.

Our locally graded category $\MainCategory$ (Def.~\ref{def:imp})
is such that its dual $(\MainCategory)\op$ can be regarded as a model of this graded algebraic theory, in the sense of \cite[Def.~11]{kura-graded-algebraic} (mildly generalized to use locally graded categories instead of actions).
We formulate the model by interpreting $(\probplus)$ and $(\ndplus_a)$ using
\lstinline|bernoulli| and \lstinline|knight| from \S\ref{sec:maincategory}.
All the equations~\eqref{eqn:commgraded}--\eqref{eqn:assocgraded} are satisfied in this model.

One view is that we have side-stepped the Eckmann-Hilton-like obstacle~\eqref{eqn:eckmann-hilton} by requiring the explicit regrading in the symmetry law $x\ndplus_a y = c_{\neg}(y\ndplus_a x)$. 

\paragraph{Aside.} We note that a named binary choice already appears in the probabilistic setting in~\cite{beta-bernoulli}, where it has a different intuitive meaning: there, $a$ stands for an urn but not a specific draw, and $?_a$ denotes sampling from urn~$a$ according to Polya's scheme: replace with two copies of what was drawn. This is different because each $(?_a)$ can be used multiple times, gathering statistics about an urn, whereas a Knightian draw can only be used once.

\section{Summary and outlook}

We have shown that by taking a graded perspective and naming Knightian choices we can obtain a compositional account of Bernoulli and Knightian uncertainty together, satisfying the Desiderata of \S\ref{sec:intro:desiderata}. The account gives a refined bound on the uncertainty (Theorem~\ref{thm:unc-bounds}) and is maximal among the compositional accounts (Theorem~\ref{thm:maximality}).

There are several future directions. An initial question is how to accommodate iteration. The convex sets considered in this article are all finitely generated. If we allow iterative programs that have an unbounded number of Knightian choices, this leads to a more general class of convex sets.

The concerns about iteration hold even if we restrict to finite outcome spaces, and thus far we have focused on this for simplicity.
Much work on programming semantics for imprecise probability has focused beyond finite outcome spaces, and it will be interesting to revisit this from our perspective: this includes domain theoretic structures (e.g.~\cite{keimel-plotkin,varacca-winskel,keimel-mfps,gl-previsionsB})
and metric structures (e.g.~\cite{mio-vignudelli,mio-sarkis-vignudelli}).

It would be interesting to compare to another recent compositional framework combining unknowns with probability by Stein and Samuelson, currently focusing on Gaussians~\cite{stein-samuelson-journal}.

Our approach is based on random elements, and so is the quasi-Borel-space probability monad (e.g.~\cite{qbs,DBLP:journals/pacmpl/VakarKS19}), so this might be a good approach to accommodating function spaces. On the other hand, we are enriching in $[\FinStochSurj\op,\Set]$, which seems closely related to
the toposes recently used in \cite{li-nominal-prob-sep,simpson-topos} for program logic and independence of random variables; in particular, \cite{li-nominal-prob-sep} considers sheaves on $\FinSetSurj$. 

On the more practical side, an open question is how to perform statistical inference in a probabilistic programming language with imprecise probability. 

Going beyond statistics, there may be other scenarios where this approach is useful: making a theory compositional by using a graded theory (for a first purely speculative example, the issues with amb outlined in~\cite{pbl-amb}).

\paragraph{Acknowledgements}
The 2023/2024 Aria workshops emphasised aspects of imprecise probability as discussion topics, and there it was helpful to discuss with davidad, Alexander Appel, Dylan Braithwaite, Vanessa Kosoy, Elena Di Lavore, Alex Lew, Owen Lynch, Sean Moss, Zenna Tavares, and others. It has been helpful to discuss with colleagues in Oxford (particularly Jeremy Gibbons and Paolo Perrone, and some time ago Kwok Ho Cheung regarding his thesis work~\cite{diaa}), and at the 2024 Bellairs workshop. We also received helpful feedback from Hugo Paquet, Dario Stein, and MFPS 2024 reviewers, where this material was presented as an `early announcement'.  

Research supported by the Clarendon Scholarship, ERC Consolidator Grant BLAST, AFOSR Project FA9550-21-1-0038, and the ARIA programme on Safeguarded AI.

\bibliographystyle{ACM-Reference-Format}
\bibliography{refs}
\end{document}